\documentclass[twocolumn,twoside]{IEEEtran}
\usepackage{times} 
\usepackage{epsfig}
\usepackage{graphicx}
\usepackage{epstopdf}
\usepackage{color,subfigure, latexsym, amsmath, amsfonts, amssymb,cite}
\usepackage{algorithm}
\usepackage{algorithmic}
\usepackage{epsfig}
\usepackage{graphicx}
\usepackage{epstopdf}
\usepackage{enumitem}
\usepackage{bm,amsmath,amssymb,amsfonts,graphicx,epsfig,amsthm,color, xcolor}
\usepackage{bbold,dsfont}
\usepackage{float}
\usepackage[switch]{lineno}
\usepackage[hyphens]{url}
\PassOptionsToPackage{bookmarks=false}{hyperref}

\usepackage[depth=-1]{bookmark}

\usepackage{booktabs}       
\usepackage{amsfonts}       
\usepackage{microtype}      

\usepackage{times}
\usepackage{graphicx} 

\usepackage{wrapfig}



\usepackage{accents}
\newcommand{\R}{\mathbb{R}}

\newcommand{\X}{\mathbb{X}}

\def\be{\begin{equation}}
\def\ee{\end{equation}}
\def\ba{\begin{array}}
\def\ea{\end{array}}
\newcommand{\smat}[1]{\left[\begin{smallmatrix}#1\end{smallmatrix}\right]}
\newcommand{\bmat}[1]{\begin{bmatrix}#1\end{bmatrix}}

\newtheorem{lemma}{Lemma}

\newtheorem{theorem}{Theorem}
\newtheorem{definition}{Definition}
\newtheorem{assumption}{Assumption}

\theoremstyle{remark}\newtheorem{remark}{Remark}

\allowdisplaybreaks

\title{\LARGE \bf
 Data-driven Kernel-based Predictive Control with Stability\\ and Robustness Guarantees
}
\author{Wenjie~Liu,~Yifei~Li, Gang Wang,~\IEEEmembership{Senior Member,~IEEE}, and Lihua Xie,~\IEEEmembership{Fellow,~IEEE}
		\thanks{
        This work was supported by the Ministry of Education, Singapore, under AcRF Tier 1 Grant RG64/23, and in part by the National Natural Science Foundation of China under Grant U23B2059.
        }
		\thanks{
Wenjie Liu, Yifei Li, and Lihua Xie are with the Centre for Advanced Robotics Technology  Innovation (CARTIN), School of Electrical and Electronic Engineering, Nanyang Technological University, Singapore (e-mail:
 (wenjie.liu, li.yifei, elhxie)@ntu.edu.sg).
Gang Wang is with the State Key Lab of Autonomous Intelligent Unmanned Systems and the School of Automation, Beijing
Institute of Technology, Beijing 100081, China (email: gangwang@bit.edu.cn).
			}
	}

\begin{document}
	
	\maketitle

 \allowdisplaybreaks

 


\begin{abstract}                
In this paper, we provide a theoretical analysis of the closed-loop properties of a data-driven kernel-based predictive control (DDKPC) scheme developed solely from input-output data. 
The proposed formulation integrates a robust data-driven predictive control framework with a multi-step predictor for nonlinear systems constructed via kernel-based methods. 
This predictor implicitly captures the system's nonlinear behavior using the representer theorem. 
For the nominal case with noise-free data, we prove that the DDKPC scheme guarantees recursive feasibility and closed-loop stability, provided that the prediction horizon is sufficiently long and the kernel representation error is sufficiently small. 
To facilitate real-time implementation, we introduce a penalty relaxation formulation to alleviate the computational burden inherently caused by nonconvex implicit constraints.
Furthermore, the framework is robustified against measurement noise by aggregating the representation mismatch and the bounded noise into a unified uncertainty bound. 
Finally, we extend the DDKPC framework to slowly time-varying
nonlinear systems by periodically reconstructing the kernel predictor
from a fixed-budget online dictionary managed by the approximate
linear dependency (ALD) criterion. 
Under suitable conditions on the rate of variation of the input-output
evolution and the online prediction error, recursive feasibility and
practical closed-loop stability are preserved.
The effectiveness of the proposed approach is illustrated through numerical examples.

\end{abstract}



\section{Introduction}
Model predictive control (MPC) is widely used because it systematically computes control actions through online optimization while handling nonlinear dynamics, constraints, and closed-loop performance requirements \cite{rawlings2019model,zhou2023ral}. However, its effectiveness relies on an accurate predictive model, whose construction can be costly or infeasible for systems with high dimensionality, complex couplings, and significant unmodeled effects. This modeling burden remains a key obstacle to broader MPC deployment.


The increasing volume of data available from modern systems, combined with advances in computation, has motivated a shift toward data-driven predictive control frameworks that reduce or eliminate explicit model identification.
For linear time-invariant (LTI) systems, \textit{Willems et al's} fundamental lemma \cite{willems2005note} provides a direct data-based parametrization of admissible trajectories, which leads to recent advancement in data-driven control \cite{persis2020data,zhao2025ddlqr,wang2026koopman,yuan2026dd,wei2025dd,eising2024data}.
In this category, data-driven predictive control (DDPC) frameworks have shown that stabilizing receding-horizon controllers can be synthesized purely from input-output data, even in settings involving disturbances and measurement noise \cite{Coulson2019data,berberich2019data,hu2025robustddmpc,shinohara2026data}, and network-induced issues \cite{Liu2023data,Liu2023self,wang2025ddreview}.
However, these results lie in the category of LTI systems, while many emerging engineering systems operate in regimes where nonlinear effects are intrinsic and cannot be neglected.

Recent research has therefore focused on extending DDPC framework to nonlinear settings.
Efforts along this direction include data-driven MPC based on local linearization \cite{berberich2021linearII}, feedback linearization \cite{alsalti2023fb-mpc}, and nonlinear lifting using neural networks \cite{lazar2024neural-ddmpc}, kernels \cite{huang2023robust,azarbahram2024data}, or Koopman operators \cite{xiong2025data}.
While these methods demonstrate promising empirical performance and broaden applicability beyond LTI systems, existing kernel- and neural-network-based approaches generally lack guarantees such as recursive feasibility or closed-loop stability.
This gap arises because terminal sets and terminal costs are typically not incorporated, and lower bounds on the prediction horizon ensuring stability remain unknown.

To bridge these theoretical and practical gaps, this paper develops a DDKPC framework for unknown nonlinear systems equipped with formal closed-loop guarantees. 
At its core, the proposed approach leverages the representer theorem \cite{scholkopf2001generalized} to construct a multi-step implicit predictor directly from historical input-output trajectories. Based on this formulation, we establish the recursive feasibility and practical closed-loop stability of the nominal scheme without relying on conventional terminal ingredients. 
Recognizing the computational bottlenecks of nonconvex optimization and the inevitability of measurement noise, we subsequently extend the baseline framework. We adopt a penalty relaxation strategy to enable efficient first-order optimization, and encapsulate bounded noise and representation mismatches within a unified uncertainty bound to guarantee robustness. 
Ultimately, we consider slowly time-varying nonlinear systems, for which a
predictor constructed from a fixed offline dataset may no longer capture
the current input-output behavior. To address this issue, we maintain an
ALD-managed fixed-budget dictionary along the closed-loop operation and
periodically update the predictor components in the
DDKPC formulation. We show that the recursive feasibility and practical
stability arguments extend to this time-varying setting under suitable
bounded-rate and prediction-error conditions.

In summary, the main contribution of this work is threefold:
\begin{itemize}
    \item [c1)] A nominal DDKPC framework is established for unknown nonlinear systems, providing provable recursive feasibility and practical stability without terminal ingredients, alongside a penalty relaxation strategy for efficient real-time computation;
    \item [c2)] The closed-loop theoretical guarantees are systematically extended to accommodate measurement noise by aggregating data uncertainty and representation mismatch into a unified bound; and, 
    \item [c3)] The DDKPC framework is extended to slowly time-varying nonlinear
systems via an ALD-managed online dictionary and periodic updates of the
implicit predictor components, with recursive feasibility and practical
stability retained under bounded-rate evolution and uniform online
prediction-error conditions.
\end{itemize}

The most relevant prior works lie in three categories:
(i) kernelized or linearization-based DDPC schemes for nonlinear systems,
e.g., \cite{huang2023robust,berberich2021linearII};
(ii) online data-driven control for time-varying systems, e.g.,
\cite{liu2023timevarying}; and
(iii) Gaussian process-based MPC (GP-MPC), e.g.,
\cite{maiworm2021onlineGPmpc,scampicchio2025gaussian}.
Key distinctions from these studies are summarized below:

\begin{itemize}
    \item \emph{Compared with \cite{huang2023robust}.}
    While \cite{huang2023robust} robustifies a kernelized DDPC formulation through a min-max design, our approach introduces a slack variable to account for the mismatch between the kernel-based multi-step predictor and the true system evolution. More importantly, we establish recursive feasibility and practical closed-loop stability, and further extend the analysis to slowly time-varying nonlinear systems.

    \item \emph{Compared with \cite{berberich2021linearII}.}
    The method in \cite{berberich2021linearII} exploits online data to construct local linear approximations of input-affine nonlinear systems via Willems et al.'s fundamental lemma \cite{willems2005note}. In contrast, our framework addresses general nonlinear input-output dynamics through a kernel-based implicit multi-step predictor, with the analysis relying on representer theory in \cite{scholkopf2001generalized} rather than local trajectory parametrization.

    \item \emph{Compared with \cite{liu2023timevarying}.}
    Although our online scheme also uses periodic updates, as in \cite{liu2023timevarying}, the problem setting and analysis are different. 
    Specifically, \cite{liu2023timevarying} considered unknown linear time-varying (LTV) state-space systems and periodically updates stabilizing feedback gains, whereas we consider slowly time-varying nonlinear input-output systems and periodically update the implicit predictor components entering the DDKPC problem.

    \item \emph{Compared with GP-MPC.}
    Since GP regression is also built on covariance kernels, our framework is conceptually related to GP-MPC \cite{maiworm2021onlineGPmpc,scampicchio2025gaussian}. However, standard GP-MPC schemes typically learn probabilistic dynamics models and require uncertainty propagation across the prediction horizon. In contrast, we use deterministic kernel regression to construct a multi-step implicit input-output predictor, avoiding recursive propagation through unknown nonlinear dynamics. Moreover, the proposed formulation does not require full state measurements or conventional terminal ingredients, while still providing recursive feasibility and practical stability guarantees.
\end{itemize}

A preliminary conference version of this paper was submitted to IFAC conference \cite{liu2026ifac}. 
Compared with the conference version, this paper makes three substantial extensions. 
First, it provides a detailed proof of the nominal recursive feasibility and practical stability result, which was only outlined in the conference version due to space limitations. 
Second, it develops a computationally tractable robust DDKPC formulation for real-time and noisy-data implementation. 
Third, it extends the framework to slowly time-varying nonlinear systems via an ALD-managed fixed-budget online dictionary and periodic predictor updates, while preserving recursive feasibility and practical stability.

\emph{Notation.} The set of real numbers is denoted by $\mathbb{R}$, and $\mathbb{R}_{>0}$ denotes the set of positive real numbers. The sets of non-negative and positive integers are denoted by $\mathbb{N}$ and $\mathbb{N}_{+}$, respectively. The sets of real symmetric positive semidefinite and positive definite matrices in $\mathbb{R}^{n \times n}$ are written as $\mathbb{S}^{n }$ and $\mathbb{S}^{n}_{>0}$.
For a vector $x \in \mathbb{R}^{n_x}$, $\|x\|$ denotes its Euclidean norm, and for a matrix $M$, $\|M\|$ denotes its spectral norm. For matrices $P_1 \in \mathbb{S}^{n_1}, \ldots, P_k \in \mathbb{S}^{n_k}$, we denote the smallest and largest eigenvalues among them by $\lambda_{\min}(P_1,\ldots,P_k)$ and $\lambda_{\max}(P_1,\ldots,P_k)$, respectively.
For vectors $V_0 \in \mathbb{R}^{n_0}, \ldots, V_k \in \mathbb{R}^{n_k}$, we use ${\rm col}(V_0,\ldots,V_k)$ to denote the stacked vector $\smat{V_0^\top & \cdots & V_k^\top}^\top$. A stacked window of a sequence $\{x_t\}_{t=t_1}^{t_2}$ is written as
$x_{[t_1,t_2]} := \smat{x_{t_1}^\top & \cdots & x_{t_2}^\top}^\top$.
For a vector $x \in \mathbb{R}^{n}$, $x_{[i]}$ denotes its $i$-th component, and $x_{[i:j]} := [x_{[i]}~\cdots~x_{[j]}]^\top$. For a matrix $A \in \mathbb{R}^{m \times n}$, $A_{[:,i]}$ denotes its $i$-th column, $A_{[i,:]}$ its $i$-th row, and $A_{[i,j]}$ its $(i,j)$-th entry. The vector $\mathbf{1}_a \in \mathbb{R}^{1 \times n}$ denotes the row vector whose $a$-th element is $1$ and all other elements are zero.
For functions $f(\cdot)$ and $h(\cdot)$, $h \circ f$ denotes the composition $h(f(\cdot))$, and $f^{[k]} := \underbrace{f \circ \cdots \circ f}_{k~\text{times}}$ denotes the $k$-fold composition of $f$.
Given $N,L \in \mathbb{N}_+$, the Hankel matrix of order $L$ of a length-$N$ trajectory $x_{[0,N -1]}$ is denoted by $H_{L}(x_{[0,N -1]}):=\smat{
	x_{0} & x_{1} & \ldots & x_{N-L} \\
	\vdots & \vdots & \ddots & \vdots \\
	x_{L-1} & x_{L} & \ldots & x_{N-1}
	}$.

\section{Preliminaries and Problem Formulation}
This section begins with a brief review of kernel-based function representation, followed by the problem formulation and the development of a data-driven kernel-based predictor.

\subsection{Kernels and their RKHS}
\label{sec:pre:ker-def}

We first recall the definition of a positive semidefinite kernel.
\begin{definition}
{\cite[Definition 1]{maddalena2021deterministickernel}}.
\label{def:positive-kernel}
Let $\Omega \subset \mathbb{R}^n$ be a nonempty set, $N \in \mathbb{N}_+$, and consider pairwise distinct points $\{x_0,\dots,x_{N-1}\} \subseteq \Omega$.  
A continuous function $K:\Omega \times \Omega \to \mathbb{R}$ is called a \emph{positive semidefinite kernel} on $\Omega$ if $\sum_{i=0}^{N-1} \sum_{j=0}^{N-1} w_i w_j K(x_i,x_j) \ge 0$
for any nonzero set of weights $\{w_0,\dots,w_{N-1}\} \subset \mathbb{R}$.

\end{definition}
Next, we recall the notion of a reproducing kernel.
\begin{definition}
    {\cite[p. 343]{aronszajn1950theory-reproducing-ker}}.
    \label{def:reproducing-kernel}
    Let $\mathcal{H}$ be a real Hilbert space of functions $g:\Omega \to \mathbb{R}$.  
A function $K:\Omega \times \Omega \to \mathbb{R}$ is a \emph{reproducing kernel} of $\mathcal{H}$ if:
\begin{itemize}
    \item[i)] For every $y \in \Omega$, the function $K(\cdot,y)$ belongs to $\mathcal{H}$.
    \item[ii)] For every $y \in \Omega$ and every $q \in \mathcal{H}$, $ q(y) = \langle q(\cdot), K(\cdot,y) \rangle_{\mathcal{H}}$,
    where $\langle \cdot,\cdot \rangle_{\mathcal{H}}$ denotes the inner product in $\mathcal{H}$.  
    This property is called the \emph{reproducing property}.
\end{itemize}
\end{definition}
The Moore-Aronszajn theorem \cite[p.~344]{aronszajn1950theory-reproducing-ker} states that any positive semidefinite and symmetric kernel $K$ uniquely determines a Hilbert space for which $K$ is the reproducing kernel.  
Such a space is called a \emph{reproducing kernel Hilbert space} (RKHS).
Hence, from Definition~\ref{def:reproducing-kernel}, for any $y \in \Omega$ the function $K(\cdot,y)$ belongs to the RKHS~$\mathcal{H}$.  
Moreover, for any finite expansion
$q(\cdot) = \sum_{i = 0}^{N-1} w_i K(\cdot,x_i)$ where $N \in \mathbb{N}_+$, $w_i \in \R$ and $x_i \in \Omega$, we have that $q \in \mathcal{H}$ and its RKHS function norm is $\|q\|_{\mathcal{H}} := \sqrt{\langle q,q \rangle}_{\mathcal{H}}$.
Furthermore, $ \|q\|^2_{\mathcal{H}} = \sum_{i = 0}^{N - 1} \sum_{j = 0} ^{N - 1} w_i w_j K(x_i,x_j) = w ^\top \bar{\textbf{K}} w$ where $w = [w_0~\cdots~w_{N - 1}]^\top$ and $\bar{\textbf{K}} \in \mathbb{S}^{N \times N}$ is the Gram matrix with $\bar{\textbf{K}}_{i,j} = K(x_i,x_j)$.

\subsection{Representer Theorem}
\label{sec:pre:ker-repre}
The RKHS concepts introduced above show how kernels $K$ induce structured function spaces for representing nonlinear mappings $q(\cdot)$. Building on this foundation, we now revisit the representer theorem, which plays an important role in this paper.

Consider a continuous positive semidefinite reproducing kernel $K:\Omega \times \Omega \rightarrow \mathbb{R}$ with RKHS~$\mathcal{H}$, and let $q:\Omega \to \mathbb{R}$ be an unknown function belonging to $\mathcal{H}$.  
Suppose $q$ generates the data points $(x_i, y_i)$, $i=0,\dots,N-1$, where $y_i = q(x_i)$.  
The goal is to compute an estimator $q^{ker}$ that minimizes
\begin{equation}\label{eq:pre:reg-cost}
    \sum_{i=0}^{N-1} \|y_i - q^{ker}(x_i)\|^2 
    + \gamma \|q^{ker}\|_{\mathcal{H}}^2,
\end{equation}
where $\gamma > 0$ is a regularization parameter.

By the representer theorem \cite[Theorem~1]{scholkopf2001generalized}, the minimizer takes the form $q^{ker}(x) = w \textbf{k}(x)$ where $\textbf{k}(x) := \bmat{K(x,x_0) & K(x,x_1) & \cdots & K(x,x_{N - 1})}^{\top}$ and $w \in \R^{1\times N}$ is the coefficient row vector.
The functions $K(x,x_i)$ are called kernel-based basis functions that are the kernels centered at the data points $x_i$, $i = 0,\cdots, N - 1$.
The number of kernel-based basis functions is equal to the number of  data points, and when the dataset is fixed, determining $q^{ker}(x)$ is equivalent to computing the coefficients $w$.

It follows from~\cite{maddalena2021deterministickernel} that
\begin{equation}\label{eq:pre:ker-predictor}
    q^{ker}(x)
    = Y (\gamma I + \bar{\mathbf{K}})^{-1} \mathbf{k}(x)
\end{equation}
where $Y = \bmat{y_0 & y_1 & \cdots & y_{N - 1}}$
and $\bar{\textbf{K}}$ is the Gram matrix.
The estimator~\eqref{eq:pre:ker-predictor} admits the following deterministic finite-sample error bound. 
\begin{lemma}{\cite[Lemma 1]{lahr2025kernelbound}}.
\label{lem:pre:kernel-est-error}
    Suppose $q \in \mathcal{H}$, where $\mathcal{H}$ is the RKHS associated with a continuous positive semidefinite kernel $K:\Omega \times \Omega \to \mathbb{R}$, and assume an a priori bound $\|q\|_{\mathcal{H}} \le \Gamma$.  
Then, for all $x \in \Omega$,
\begin{align}\label{eq:pre:kernel-error}
    |q^{ker}(x) - q(x)|
    &\le 
    \sqrt{\,\Gamma^2 - Y(\bar{\mathbf{K}} + \gamma I)^{-1} Y^\top\,} \\
    &\hspace{0.7cm}\times 
    \sqrt{\,K(x,x) - 
    \mathbf{k}(x)^\top (\bar{\mathbf{K}} + \gamma I)^{-1} \mathbf{k}(x)\,}. \nonumber
\end{align}
\end{lemma}

With these preliminaries in place, we are now ready to formulate the problem addressed in this paper.

\subsection{Problem formulation}
\label{sec:pre:prob-form}
Consider the following nonlinear discrete-time system
\begin{subequations}\label{eq:sys:x-y}
    \begin{align}
        x_{t + 1} &= f_0(x_t,u_t)\\
        y_t &= b_0(x_t) 
    \end{align}
\end{subequations}
where $x\in \R^{n}$, $u \in \R^m$, and $y \in \R^p$ are the state, input, and output, respectively.
Functions $f_0(\cdot,\cdot)$ and $b_0(\cdot)$ are unknown.
We assume that $(x^s,u^s)$ is a known unstable equilibrium of the system, and that its output satisfies $y^s = b_0(x^s)$.
The objective is to design a DDPC framework, based only on input-output data, that stabilizes the system dynamics around the equilibrium.
To this end, we first recharacterize system \eqref{eq:sys:x-y} in terms of its input-output behavior, leading to the following assumption.
\begin{assumption}
\label{as:io-equiv}
    There exists a known integer $\eta >0$ such that the state  
    \begin{equation}
\label{eq:xi}
\xi_t = \smat{
u_{t - \eta}^\top &
\cdots & u_{t - 1}^\top &
y_{t - \eta}^\top &
\cdots & y_{t - 1}^\top
}^\top \in \R^{n_{\xi}}
    \end{equation}
with dynamics
\begin{subequations}\label{eq:sys:xi-y}
        \begin{align}
            \xi_{t + 1} &= f(\xi_t,u_t) \label{eq:sys:xi-y:xi}\\
        y_t &= b(\xi_t) \label{eq:sys:xi-y:y}
        \end{align}
    \end{subequations}
    has the same input-output behavior as the system \eqref{eq:sys:x-y}.
\end{assumption}
\begin{remark}
\label{rmk:io-equiv}
    When system \eqref{eq:sys:x-y} is a linear time-invariant (LTI) system, i.e.,
\begin{subequations}\label{eq:sys:lti:x-y}
    \begin{align}
        x_{t + 1} & = Ax_t+ Bu_t\\
        y_t & = Cx_t + D u_t ,
    \end{align}
\end{subequations}
Assumption \ref{as:io-equiv} is satisfied with $\eta \ge \eta_0$ where $\eta_0$ is the observability index of the system \eqref{eq:sys:lti:x-y}.
In the nonlinear case, a realization of the form \eqref{eq:sys:xi-y} 
exists under suitable observability and finite relative degree 
conditions on system \eqref{eq:sys:x-y}, which ensure that the mapping 
from the state $x_t$ to a finite input-output segment is locally 
injective \cite[Chapters 3-4]{isidori1995nonlinear}. 
Such finite-memory input-output representations are classical in 
nonlinear system theory and include, e.g., nonlinear autoregression 
with exogenous inputs (NARX) models and systems with well-defined 
relative degree. 
While these conditions do not hold universally, they are satisfied 
by a broad class of nonlinear systems encountered in practice.
\end{remark}

Apart from Assumption~\ref{as:io-equiv}, our knowledge of~\eqref{eq:sys:x-y} is based on data collected in an experiment on~\eqref{eq:sys:x-y}.
The experiment returns, for some integer $N \ge \eta$, a sequence of input-output data $\mathcal{D}: = \{u^{\rm d}_{[0,N - 1]}$, $y^{\rm d}_{[0,N - 1]}\}$ and their associated Hankel matrix $H_{L + \eta}(u^{\rm d}, y^{\rm d})$. 
Here, the superscript ``d" denotes data collected offline.
Let $L \in \mathbb{N}_+$ be the prediction horizon, and
$N_c := N - \eta - L + 1$.
Partition the Hankel matrix $H_{L + \eta}(u^{\rm d}, y^{\rm d})$ into 
\begin{equation}
    \smat{
    U_p\\
    U_F
    } := H_{\eta + L}(u^{\rm d}),~\smat{
    {Y}_p\\
    {Y}_F
    } := H_{\eta + L}(y^{\rm d})
\end{equation}
where $U_P \in \R^{m \eta \times N_c}$, $U_F \in \R^{m L \times N_c}$, ${Y}_P \in \R^{p \eta \times N_c}$, ${Y}_F \in \R^{pL \times N_c}$.
For the LTI system \eqref{eq:sys:lti:x-y}, it has been shown by the fundamental lemma in \cite{willems2005note} that if $u^{\rm d}$ is persistently exciting with $L + \eta + n$ order, then rank($H_{L + \eta}(u^{\rm d},y^{\rm d})) = m(\eta + L) + n$ is always satisfied.
Under this condition, ${\rm col}(u_{ini},\bar{u},{y}_{ini}, \bar{y}) \in \R^{(m + p)(\eta + L)}$ is a trajectory of \eqref{eq:sys:x-y} if and only if there exists a vector $g \in \R^{N_c}$, such that 
\begin{equation}\label{eq:data-repre-LTI}
    \smat{
    U_P\\
    {Y}_P\\
    U_F\\
    {Y}_F
    } g = \smat{
    u_{ini}\\
    {y}_{ini}\\
    \bar{u}\\
    \bar{y}
    }.
\end{equation}
According to Assumption~\ref{as:io-equiv}, the initial trajectory $\xi_{ini}:={\rm col}(u_{ini}, y_{ini}) \in \mathbb{R}^{(m+p)\eta}$ can be viewed as setting the initial condition for the future (to-be-predicted) trajectory ${\rm col}(\bar{u},\bar{y}) \in \mathbb{R}^{(m+p)L}$. In other words, given $(u_{ini}, y_{ini})$, every future input trajectory $\bar{u}$ uniquely determines the corresponding future output trajectory $\bar{y}$ through~\eqref{eq:data-repre-LTI}. We refer interested readers to \cite{berberich2019data,Liu2023data,Coulson2019data} for detailed implementations of DDPC schemes for LTI systems.

However, the fundamental lemma holds only for LTI systems, which limits the practical applicability of existing DDPC schemes, as real-world systems are typically nonlinear. Motivated by this limitation, our aim is to develop a new DDPC scheme based on the data $(u^{\rm d}_{[0,N-1]}, y^{\rm d}_{[0,N-1]})$ that stabilizes the nonlinear system~\eqref{eq:sys:x-y} while providing performance guarantees. To this end, we resort to the kernel-based representation reviewed in Sections \ref{sec:pre:ker-def} and \ref{sec:pre:ker-repre} and replace~\eqref{eq:data-repre-LTI} with a kernel-based representation.

\subsection{Kernel-based multi-step predictor}
\label{sec:pre:multi-step-predictor}
Given an initial condition $(u_{ini}, y_{ini})$ and an input sequence 
$\bar{u} \in \mathbb{R}^{mL}$, we predict, for each $a \in [1,pL]$, the $a$-th 
element of the future output sequence $\bar{y} \in \mathbb{R}^{pL}$ generated 
by the nonlinear system~\eqref{eq:sys:x-y} using the multi-step predictor
\begin{equation}\label{eq:multi-step-predictor}
    \bar{y}_{[a]} = q_{[a]}(u_{ini}, y_{ini}, \bar{u}_{[1:mL]}).
\end{equation}
Here, $\bar{y} = \bmat{\bar{y}_{[1]} & \cdots & \bar{y}_{[pL]}}^\top$ is the 
predicted output. We first justify the structure in 
\eqref{eq:multi-step-predictor} by showing that, when $q_{[a]}$ is chosen 
appropriately, $\bar{y}_{[a]} = y_{[a]}$ holds for all $a \in [1,pL]$.

Specifically, by iterating the dynamics~\eqref{eq:sys:xi-y}, one obtains
\begin{subequations}\label{eq:y-iter0}
    \begin{align}
        &\bar{y}_{[1:p]} = b(\xi_{ini}), \\
        &\bar{y}_{[p+1:2p]} = b \circ f(\xi_{ini}, \bar{u}_{[1:m]}), \\
        &\qquad \cdots \nonumber \\
        &\bar{y}_{[(L-1)p+1:pL]} 
            = b \circ f^{[L-1]}(\xi_{ini}, \bar{u}_{[1:mL]}).
    \end{align}
\end{subequations}
Under the standard causality assumption (the output at time $t$ is 
independent of inputs $u_{t+1}, u_{t+2},\dots$), each element satisfies, 
for $a \in [1,p]$,
\begin{subequations}\label{eq:y-iter1}
    \begin{align}
        &\bar{y}_{[a]}  
            = \mathbf{1}_a b(\xi_{ini}, \bar{u}_{[1:mL]}) 
            \triangleq q_{[a]}(\xi_{ini}, \bar{u}_{[1:mL]}), \\
        &\bar{y}_{[a+p]}  
            = \mathbf{1}_a b \circ f(\xi_{ini}, \bar{u}_{[1:mL]}) 
            \triangleq q_{[a+p]}(\xi_{ini}, \bar{u}_{[1:mL]}), \\
        &\qquad \cdots \nonumber \\
        &\bar{y}_{[a+(L-1)p]}  
            = \mathbf{1}_a b \circ f^{[L-1]}(\xi_{ini}, \bar{u}_{[1:mL]}) \\
        &~~\qquad \qquad \triangleq 
            q_{[a+(L-1)p]}(\xi_{ini}, \bar{u}_{[1:mL]}).
    \end{align}
\end{subequations}
Thus, the predictor~\eqref{eq:multi-step-predictor} exactly reproduces the 
true output of~\eqref{eq:sys:x-y} when the functions $q_{[a]}$ coincide with 
the compositions of $f$ and $h$ as in~\eqref{eq:y-iter1}. However, because 
$f$ and $h$ are unknown in our data-driven setting, directly constructing 
$q_{[a]}$ is infeasible.

Therefore, 
we seek a data-driven multi-step predictor that maps $ {\rm col}(\xi_{ini}, \bar{u}_{[1:mL]})$
to the future output sequence. For the subsequent closed-loop analysis, 
such a predictor should not only be computable from data, but also admit 
a tractable characterization of its approximation error. Motivated by 
these requirements, we adopt a kernel-based representation, since it 
yields a finite-dimensional predictor through the representer theorem 
and admits an explicit estimation error bound. Specifically, we 
approximate the functions $q_{[a]}$ using kernel-based representations 
fitted to the observed data contained in ${\rm col}(U_P, Y_P, U_F, Y_F)$. 
To this end, we impose the following assumption, which is commonly used in, e.g., \cite{hu2023learning}.

\begin{assumption}[Kernel representation]
\label{as:kernel-representation}
    For every $a \in [1,pL]$, the function $q_{[a]}$ belongs to an RKHS 
    $\mathcal{H}$ induced by a continuous positive semidefinite kernel 
    $K : \Omega \times \Omega \to \mathbb{R}$ with 
    $\Omega \subset \mathbb{R}^{(m+p)\eta + mL}$. Moreover, $\|q_{[a]}\|_{\mathcal{H}} \le \Gamma$
    for some known constant $\Gamma > 0$.
\end{assumption}

Letting 
$\zeta := {\rm col}(u_{ini}, y_{ini}, \bar{u}_{[1:mL]}) \in \Omega$ 
and, for $j \in [1,N_c]$,  
$\zeta^{\rm d}_j := {\rm col}(U_{P[:,j]}, Y_{P[:,j]}, U_{F[:,j]})$, 
the kernel-based estimate of the output is given by
\begin{equation}\label{eq:kernel-estimator}
    y^{ker}_{[1:pL]}
        = Y_F (\bar{\mathbf{K}} + \gamma I)^{-1}
          \mathbf{k}(\zeta),
\end{equation}
where $\gamma > 0$ is a regularization parameter,
$\mathbf{k}(\zeta) = [K(\zeta^{\rm d}_1,\zeta) $ $ \dots 
K(\zeta^{\rm d}_{N_c},\zeta)]^\top$, and the Gram matrix 
$\bar{\mathbf{K}} \in \mathbb{S}^{N_c \times N_c}$ satisfies 
$\bar{\mathbf{K}}_{[i,j]} = K(\zeta^{\rm d}_i, \zeta^{\rm d}_j)$ constructed using purely offline data.

Since the true output satisfies~\eqref{eq:multi-step-predictor} under 
$\zeta$, Assumption~\ref{as:kernel-representation} together with 
Lemma~\ref{lem:pre:kernel-est-error} yields the estimation error bound
\begin{align}\label{eq:kernel-error}
    |y^{ker}_{[a]} - y_{[a]}|
        &\le d(\zeta) \nonumber \\
        &:= \sqrt{\Gamma^2 
            - Y_{F[a,:]}(\bar{\mathbf{K}} + \gamma I)^{-1} 
              Y_{F[a,:]}^\top} \\
        &\quad \times 
        \sqrt{K(\zeta,\zeta) 
            - \mathbf{k}(\zeta)^\top
              (\bar{\mathbf{K}} + \gamma I)^{-1}
              \mathbf{k}(\zeta)}. \nonumber
\end{align}

We further impose the following boundedness assumptions on the error term 
and the kernel.

\begin{assumption}\label{as:bar_d}
    For all $\zeta \in \Omega$, there exists a known constant 
    $\bar{d} > 0$ such that $d(\zeta) \le \bar{d}$.
\end{assumption}

\begin{assumption}\label{as:bound-kernel}
    For any $\zeta_1, \zeta_2 \in \Omega$, there exists 
    $\bar{k} > 0$ such that 
    $\lvert K(\zeta_1,\zeta_2) \rvert \le \bar{k}$.
\end{assumption}

\begin{remark}\label{rmk:assumptions_kernel}
    Assumption~\ref{as:bar_d} is readily satisfied on any compact prediction
domain. Specifically, if
$(u_{ini},y_{ini})$ remains in a compact set along the closed-loop operation
and the predicted input sequence satisfies $\bar u\in\mathcal U^L$ with
compact $\mathcal U$, then
$\zeta=\mathrm{col}(u_{ini},y_{ini},\bar u)$ belongs to a compact subset of
$\Omega$. Since $d(\zeta)$ is continuous under the regularized kernel
predictor, it admits a finite upper bound on this compact set.
    Moreover, for bounded kernels such as the 
    Gaussian kernel $K(\zeta_i,\zeta_j)
        = \sigma_f^2 \exp\!\left(
            -\tfrac{1}{2}(\zeta_i - \zeta_j)^\top 
            \Upsilon (\zeta_i - \zeta_j)
          \right)$, where $\sigma_f$ and $\Upsilon$ are hyperparameters,
    Assumptions~\ref{as:bar_d} and~\ref{as:bound-kernel} hold without requiring 
    constraints on $\bar{u}$.
    Although Assumptions~\ref{as:kernel-representation}--\ref{as:bound-kernel} are standard in the literature, obtaining the exact analytical values of $\Gamma$ and $\bar{d}$ is generally intractable for unknown black-box systems. However, this does not hinder practical implementation. As shown in the subsequent theoretical analysis and the objective design, $\bar{d}$ primarily serves as a scaling factor to balance the penalty terms ($\lambda_h/\bar{d}$ and $\lambda_g\bar{d}$) when constructing feasible candidate solutions. Therefore, an exact analytical $\bar{d}$ is not strictly required in practice; knowing its approximate order of magnitude is sufficient to maintain appropriate parameter separation, which can be readily achieved via numerical tuning.
\end{remark}

Building on this kernel-based interpolation method, system \eqref{eq:sys:xi-y:xi} can be regarded as 
\begin{equation}\label{eq:f_ker+d}
    \xi_{t + 1} = f(\xi_t,u_t) =  f^{ker}(\xi_t,u_t) + d^{ker}(\xi_t)
\end{equation}
where $ f^{ker}(\xi_t,u_t) \triangleq \xi_{t + 1}^{ker}$ is formulated using $y^{ker}$ from \eqref{eq:kernel-estimator} and  $d^{ker}(\xi_t)$ represents the interpolation error satisfying
\begin{align}\label{eq:kernel-error-xi}
    &\|d^{ker}(\xi_t)\| = \|\xi_{t + 1} - \xi^{ker}_{t + 1} \| = \Big\|\smat{
    u_{[t - \eta + 1,t]}\\
    y_{[t - \eta + 1,t]}
    } - \smat{
    u_{[t - \eta + 1,t]}\\
    y^{ker}_{[t - \eta + 1,t]}
    }
    \Big\| \nonumber\\
    &=   \Big\|\smat{
    0\\
    y_{[t - \eta + 1,t]} - y^{ker}_{[t - \eta + 1,t]}
    } 
    \Big\| \overset{\eqref{eq:kernel-error}}{\le}  p d(\zeta) \le p\bar{d}.
\end{align}

So far, we have introduced the kernel-based multi-step predictor \eqref{eq:kernel-estimator}.
Before moving on to the controller design, we formalize several structural assumptions on the nonlinear system \eqref{eq:sys:x-y}. These conditions play a central role in establishing the stability and performance guarantees in the next section.
\begin{assumption}
\label{as:stab}
    There exists a continuous Lyapunov function $V_s(\xi)$, a Lipschitz continuous feedback $\kappa(\xi)$, a compact set $\mathcal{I}$, and constants $c_{s,l}, c_{s,u},\delta^{loc} >0$, $\rho_s \in (0,1)$, such that for any $\xi \in \X_{\delta} :=\{\xi|V_s(\xi) \le \delta^{loc} \}$ and $d \in \mathcal{I}$, it holds that $f(\xi,\kappa(\xi)) + d \in \mathbb{X}_{\delta}$ and 
    \begin{subequations}\label{eq:Vs-property}
        \begin{align}
            & c_{s,l}\|\xi\|^2 \le V_s(\xi) \le c_{s,u}\|\xi\|^2 \label{eq:Vs-property:l-u-bound}\\
            & V_s(f^{}(\xi,\kappa(\xi)) + d) \le \rho_s V_s(\xi) + \alpha_s(\|d\|) \label{eq:Vs-property:decrease}
        \end{align}
    \end{subequations}
    where $\alpha_s(\cdot) \in \mathcal{K}_{\infty}$ \footnote{A function $\alpha : [0,\infty) \rightarrow [0, \infty)$ is said to be of class $\mathcal{K}$ if it is continuous, strictly increasing, and $\alpha(0) = 0$.
			A function $\alpha : [0,\infty) \rightarrow [0, \infty)$ is said to be of class $\mathcal{K}_{\infty}$ if it is of class $\mathcal{K}$ and also unbounded.}.
            Moreover, when $d = 0$, it holds that $V_s(f(\xi,\kappa(\xi))) \le \rho_s V_s(\xi)$. 
\end{assumption}

The following assumption is adapted from the uniformly input-output-to-state stable (UIOSS) in \cite[Definition 3.7]{CAI2008326}.
\begin{assumption}
    \label{as:observ}
    There exists a continuous UIOSS Lyapunov function $V_o(\xi)$, a compact set $\mathcal{I}$, constants $c_{o,l},c_{o,u},\epsilon_o >0$, $\rho_o \in (0,1)$, matrices $Q \in \mathbb{S}^p$ and $R \in \mathbb{S}^m_{>0}$, such that for any $\xi$ and $d \in \mathcal{I}$, it holds that 
\begin{subequations}\label{eq:Vo-property}
    \begin{align}
        & c_{o,l} \|\xi\|^2 \le V_o(\xi) \le c_{o,u} \|\xi\|^2   \label{eq:Vo-property:l-u-bound},\\
        & V_o(f^{}(\xi,u)\! +d) \!- V_o(\xi) \le \!-\epsilon_o \|\xi \|^2 + \!\frac{1}{2}\|y\|_Q^2 + \!\|u\|_R^2 \label{eq:Vo-property:decrease}.
    \end{align}
\end{subequations}

\end{assumption}

\begin{remark}\label{rmk:local-Lya-ctrl}
    Assumptions \ref{as:stab} and \ref{as:observ} are standard in the literature for learning-based and data-driven MPC (see, e.g., \cite{berberich2021linearII}). 
    It is worth emphasizing that the local stabilizing feedback $\kappa(\xi)$ in Assumption \ref{as:stab} is used only for theoretical analysis, and need not be explicitly calculated in practice. 
    On the other hand, such local controllers $\kappa(\xi)$ and the corresponding Lyapunov function $V_s$ can be synthesized offline by several data-driven methods, e.g., Taylor expansion \cite{guo2022data}.
    Nevertheless, even when such a controller is available, it is typically not designed to directly handle hard constraints or to provide the enlarged operating region targeted by MPC. This motivates the proposed DDKPC formulation, which incorporates these features at the optimization level.
   Furthermore, the terms
$\frac{1}{2}\|y\|_Q^2+\|u\|_R^2$ on the right-hand side of
\eqref{eq:Vo-property:decrease} are not essential and can be replaced by
general $\mathcal K_\infty$ functions of $\|y\|$ and $\|u\|$.
The present quadratic form is adopted only for consistency with the stage
cost used in the DDKPC formulation.
\end{remark}

\section{Data-driven Kernel-based Predictive Control}
\label{sec:ddkmp:noise-free}
In this section, we present a DDKPC scheme to control unknown nonlinear system \eqref{eq:sys:x-y} with stability guarantees, followed by its extension in the presence of measurement noise.

\subsection{DDKPC scheme}
At time $t \ge \eta$, given offline collected input-output data $\{u^{\rm d}, y^{\rm d}\}$, past input-output measurements $u_{[t - \eta, t - 1]},y_{[t - \eta, t - 1]}$ of the nonlinear system \eqref{eq:sys:x-y}, we define the following open-loop optimal control problem.
\begin{subequations}\label{eq:ddkmpc-ideal}
    \begin{align}
        J_L(\xi_t):=  \label{eq:ddkmpc-ideal:obj}\\
        \min_{g,\bar{u},\bar{y},h} ~~&\sum_{i = 0}^{L - 1} \ell(\bar{u}_{i|t}, \bar{y}_{i|t}) + \frac{\lambda_h}{\bar{d}}\|h_t\|^2 + \lambda_g \bar{d}\|g_t\|^2 \nonumber\\
         {\rm s.t.} ~~& \bmat{
        \bar{\textbf{K}} + \gamma I\\
        Y_F
        } g_t = \bmat{
        \textbf{k}(u_{ini},y_{ini},\bar{u}_{[0,L - 1]|t})\\
        \bar{y}_{[0,L - 1]|t} + h_{[0,L - 1]|t}
        }, \label{eq:ddkmpc-ideal:ker}\\
        &\bmat{
        u_{ini}\\
        y_{ini}
        } = \bmat{
        u_{[t - \eta,t - 1]}\\
        y_{[t - \eta,t - 1]}
        },\label{eq:ddkmpc-ideal:ini}\\
        & 
        \bar{u}_{i|t} \in \mathcal{U},~~\forall i \in [0,L - 1]\label{eq:ddkmpc-ideal:h-u-constraint}.
    \end{align}  
\end{subequations}
Here, $\bar{u} \in \mathbb{R}^{m L}$ and $\bar{y} \in \mathbb{R}^{p L}$ denote the input and output trajectories predicted at time $t$, with elements $\bar{u}_{i|t} \in \mathbb{R}^m$ and $\bar{y}_{i|t} \in \mathbb{R}^{p}$ corresponding to time $i + t$, respectively.  
Matrix $\bar{\textbf{K}}$ is the Gram matrix as in \eqref{eq:kernel-estimator} constructed using offline data $\{u^{\rm d}, y^{\rm d}\}$.
To maintain feasibility of the optimization problem, a slack variable $h_t = [h_{0|t} \ \cdots \ h_{L-1|t}]^\top \in \mathbb{R}^{pL}$ is introduced to compensate for the kernel representation error $d(\zeta)$.  
Constraint \eqref{eq:ddkmpc-ideal:ker} provides an implicit form of the predictor \eqref{eq:kernel-estimator}, adapted from \cite[eq.~(14)]{huang2023robust} to enhance the feasibility of Problem \eqref{eq:ddkmpc-ideal}.  
Constraint \eqref{eq:ddkmpc-ideal:ini} uses $\eta$ input-output pairs to restrict the state $x$ at time $t$ (cf. Assumption \ref{as:io-equiv}), while constraint \eqref{eq:ddkmpc-ideal:h-u-constraint} enforces the general input constraints.
The objective function \eqref{eq:ddkmpc-ideal:obj} includes a general quadratic stage cost given by $\ell(\bar{u}_{i|t}, \bar{y}_{i|t}) :=\|\bar{u}_{i|t} - u^s\|_R^2 + \|\bar{y}_{i|t} - y^s\|^2_Q$ with $Q \in \mathbb{S}^p_{>0}$, $R \in \mathbb{S}^{m}_{>0}$, as well as regularization terms on $g_t$ and $h_t$ to account for model mismatches between the nonlinear system \eqref{eq:sys:x-y} and the kernel-based predictor \eqref{eq:kernel-estimator}. 
The coefficient structures $\lambda_g\bar{d}$ and $\lambda_h/\bar{d}$ are standard in the robust DDPC literature; see, e.g., \cite{Liu2023data, berberich2021linearII}.
At time $t$, let the optimal solution of Problem \eqref{eq:ddkmpc-ideal} be $(\bar{u}^*_{[0,L - 1]|t}, \bar{y}^*_{[0,L - 1]|t}, g_t^*, h_t^*)$, with optimal cost $J^*_L(\xi_t)$.  
As in standard MPC, Problem \eqref{eq:ddkmpc-ideal} is solved in a receding horizon fashion, summarized in Algorithm \ref{alg:ddkmpc-ideal}.

	\begin{algorithm}[h]
		\caption{DDKPC.}
		\label{alg:ddkmpc-ideal}
		\begin{algorithmic}[1]
			\STATE \textbf{Input:} past data $\mathcal{D} = \{u_{[0,N - 1]}, y_{[0,N - 1]}\}$, feasible input set $\mathcal{U}$, prediction horizon $L$, matrices $Q$, $R$, and parameters $\lambda_h$, $\lambda_g$, $\bar{d}$, $\gamma$.
            \STATE \textbf{Construct:} matrices $\bar{\mathbf{K}}$, $Y_F$, 
            \STATE At time $t$, take past $\eta$ measurements $(u_{[t - \eta,t - 1]},y_{[t - \eta,t - 1]})$ and solve Problem \eqref{eq:ddkmpc-ideal}. \label{alg:mpc1}
			\STATE Apply the input $u_t = \bar{u}^*_{0|t}$ to system \eqref{eq:sys:x-y}.
			\STATE Set $t = t +1$ and go back to \ref{alg:mpc1}.
		\end{algorithmic}
	\end{algorithm}

\begin{remark}
Compared with the strictly convex quadratic DDPC schemes in \cite{berberich2019data,Coulson2019data,Liu2023data}, the proposed DDKPC in \eqref{eq:ddkmpc-ideal} is generally nonconvex in $\bar{u}$ unless the kernel is linear in $\bar{u}$.  
This nonconvexity is an unavoidable consequence of addressing general nonlinear systems, and represents a necessary trade-off for achieving broader applicability.  
In practice, however, efficient nonlinear optimization methods, e.g., \cite{huang2023robust,wang2025optimization}, can be employed to solve \eqref{eq:ddkmpc-ideal}, making this computational challenge tractable.
\end{remark}

\subsection{Closed-loop guarantees}
Without loss of generality, we assume for the analysis that $(u^s,x^s,y^s) = (0,0,0)$.
Further, assume that Problem \eqref{eq:ddkmpc-ideal} is solvable at time $t$.
The following lemma provides an important result for the cost function.

\begin{lemma}[{\cite[Lemma 2]{liu2026ifac}}]
\label{lem:local-cost-bound}
Under Assumption \ref{as:stab}, letting $\delta_s := \frac{\delta^{loc}}{c_{s,u}}$, then for all $\| \xi \|^2 < \delta_s$, Problem \eqref{eq:ddkmpc-ideal} is feasible and there exist a constant $\gamma_s > 0$ and a function $\hat{\alpha}_s(\cdot) \in \mathcal{K}_{\infty}$ such that 
\begin{equation}\label{eq:local-cost-bound}
    J^*_L(\xi) \le \gamma_s \|\xi_t\|^2 + \hat{\alpha}_s(\bar{d}).
\end{equation}
\end{lemma}

Leveraging this lemma, we are ready to present our main result.
For the following theoretical analysis, we use the Lyapunov function candidate
\begin{equation}\label{eq:Lyp-func}
    Y_L(\xi_t) := J_L^*(\xi_t) + V_o(\xi_t).
\end{equation}
\begin{theorem}[Nominal stability guarantees]
\label{thm:stab-ideal}
    Let Assumptions \ref{as:io-equiv}--\ref{as:observ} hold.
    Then, for any constant $\bar{Y} > 0$, there exist constants $L_{\bar{Y}}, \gamma_{\bar{Y}}, \bar{d}_0 > 0$ such that for $Y_L(\xi_0) \le \bar{Y}$, $L > L_{\bar{Y}}$, $\bar{d} \le \bar{d}_{0}$, Problem \eqref{eq:ddkmpc-ideal} is feasible for all times $t >0$, and the function $Y_L$ defined in \eqref{eq:Lyp-func} fulfills
    \begin{subequations}\label{eq:YL-property}
        \begin{align}
            \epsilon_o \|\xi_t\|^2 \le Y_{L}(\xi_t) \le \gamma_{\bar{Y}} \|\xi_t\|^2 + \hat{\alpha}_s(\bar{d}) \label{eq:YL-property:l-u-bound}\\
            Y_{L}(\xi_{t + 1}) - Y_{L}(\xi_t)
 \le -a_L \epsilon_o \|\xi_t\|^2 + \alpha_Y(\bar{d}) \label{eq:YL-property:decrease}
 \end{align}
    \end{subequations}
    with $\gamma_{\bar{Y}},a_L > 0$, and $\alpha_Y(\cdot)\in \mathcal{K}_{\infty}$.
\end{theorem}
\begin{proof} 
The proof is conducted by first deriving lower and upper bounds for function $Y_L$.
Then, assuming that the Problem \eqref{eq:ddkmpc-ideal} is feasible at time $t$, we construct a candidate solution to this problem at time $t + 1$.
Finally, practical Lyapunov inequality \eqref{eq:YL-property:decrease} is established. 
For simplicity, let $\ell^*_{i|t} := \ell(\bar{u}^*_{i|t},\bar{y}^*_{i|t}) =  \|\bar{u}^*_{i|t}\|^2_R + \|\bar{y}^*_{i|t}\|^2_Q$.

    \emph{Step I.}
Note that $Y_{L}(\xi_t) = J^*_L(\xi_t) + {V}_o(\xi_t) \ge \ell^*_{i|t}  + V_o(\xi_t) \overset{\eqref{eq:Vo-property:decrease}}{\ge}\epsilon_o \|\xi_t\|^2 $, where the last inequity holds because $V_o(\xi_{t + 1}) > 0$ and \eqref{eq:Vo-property:decrease}.
Next, we derive an upper bound for $Y_L(\xi_t)$.
It follows from Lemma \ref{lem:local-cost-bound} and \eqref{eq:Vo-property:l-u-bound} that for all $\xi_t \in \X_\delta$, we have that $Y_L(\xi_t) \le (\gamma_s + c_{o,u})\|\xi_t\|^2 + \hat{\alpha}_s(\bar{d})$.
On the other hand, for $\xi_t \notin \X_{\delta}$, i.e., $\|\xi_t\|^2 > \delta_s$, we have that $Y_L(\xi_t) \le \bar{Y} = \frac{\bar{Y}}{\|\xi_{t}\|^2}\|\xi_{t}\|^2 \le \frac{\bar{Y}}{\delta_s}  \|\xi_{t}\|^2$.
Hence, the right inequality in \eqref{eq:YL-property:l-u-bound} can be derived by setting $\gamma_{\bar{Y}} := \max\{\gamma_s + c_{o,u},\frac{\bar{Y}}{\delta_s} \}$.

    \emph{Step II.}
    Assume that Problem \eqref{eq:ddkmpc-ideal} is feasible at time $t$, and denote the optimal solution by $(\bar{u}^*_{[0,L - 1]|t}, \bar{y}^*_{[0,L - 1]|t}, g_t^*, h^*_t)$.
    Denote the trajectory of the actual system \eqref{eq:sys:xi-y} generated under input sequence $\bar{u}^*_{[0,L - 1]|t}$ by $\xi_{[t + 1,t + L]}$ and $y_{[t,t + L - 1]}$.
    It follows from \eqref{eq:Vo-property:decrease} that 
    \begin{align}
        V_o(\xi_{t + L}) - V_o(\xi_t) &= \sum_{i = 0}^{L - 1} V_o(\xi_{t + i + 1}) - V_o(\xi_{t + i}) \\
        &\le  \sum_{i = 0}^{L - 1} -\epsilon_o\|\xi_{i+t}\|^2 + \frac{1}{2}\|y_{t + i}\|^2_Q + \|\bar{u}^*_{i|t}\|^2_R.\nonumber
    \end{align}
Noting that $V_o(\xi_{t + L}) >0$, it follows that
\begin{align}
    &\sum_{i = 0}^{L - 1}\epsilon_o\|\xi_{i+t}\|^2 \nonumber\\
    &\le V_o(\xi_t) + \sum_{i= 0}^{L - 1} \|\bar{u}^*_{i|t}\|^2_R + \|\bar{y}^*_{i|t} \|^2_Q + \|\bar{y}^*_{i|t} - y_{t + i}\|^2_Q \nonumber\\
    & \le V_o(\xi_t) + \sum_{i= 0}^{L - 1} \ell^*_{i|t} \nonumber\\
    &~~~ + 2 \lambda_{\max}(Q) \| (\bar{y}^*_{i|t} + h^*_{i|t}) - y_{t + i}\|^2 + 2 \lambda_{\max}(Q) \|  h^*_{i|t} \|^2 \nonumber\\
    & \le Y_{L}(\xi_t) + \underbrace{2 \lambda_{\max}(Q) (Lp \bar{d}^2 + \bar{Y} \bar{d}/ \lambda_h)}_{ \triangleq \alpha_1(\bar{d})} \label{eq:alpha_1}
\end{align}
where the last inequality follows from the fact that $\sum_{i = 0}^{L - 1} \|h^*_{i|t}\|^2 \le (\bar{d}/\lambda_h) J^*(\xi_t) \le (\bar{d}/\lambda_h)\bar{Y}$ and $\bar{y}^*_{[0,L-1]|t} + h^*_{[0,L-1]|t} = Y_F g^*_t$ is a trajectory of the kernel-based estimator \eqref{eq:kernel-estimator}, and hence the estimation error satisfies Assumption \ref{as:bound-kernel}.
Noting that $Y_{L}(\xi_t) \le \bar{Y}$, there exists at least one $i_{\xi} \in \{1,\cdots,L - 1\}$ such that
\begin{align}
    \|\xi_{i_{\xi} + t}\|^2  &\le \smat{\frac{\min\{ \bar{Y} + \alpha_1(\bar{d}),\gamma_{\bar{Y}} \|\xi_t\|^2 + \alpha_1(\bar{d}) + \hat{\alpha}_s(\bar{d})\}}{\epsilon_o(L - 1)}}  \label{eq:xi-mid}\\
    &\le \frac{\bar{Y} + \theta }{\epsilon_o(L - 1)} \nonumber
\end{align}
where $\theta > 0$ is an arbitrary but fixed constant satisfying $\theta \ge \alpha_1(\bar{d}) + \hat{\alpha}_s(\bar{d})$.
Hence, if $L > L_0 := 1 + \frac{\bar{Y}}{\delta_s \epsilon_o}$, then there exists a small enough constant $\theta > 0$ such that \eqref{eq:xi-mid} ensures $\|\xi_{t + i_\xi}\|^2 \le \delta_s$ with $\delta_s$ from Lemma \ref{lem:local-cost-bound}.
Building on this result, for $i \in [0,i_\xi - 2]$, let $\bar{u}_{i|t + 1} := \bar{u}^*_{i + 1|t}$, and for $i \in [i_\xi - 1, L - 1]$, let $\bar{u}_{i|t + 1} := \kappa(\xi_{_{t + i + 1}})$.
For $i \in [0,L - 1]$, let $\bar{y}_{i|t + 1} := y_{t + i + 1}$ and $h_{[0,L - 1]|t + 1} = y^{ker}_{[t + 1, t + L]} - y_{[t + 1, t + L]}$ where $y^{ker}_{[t + 1, t + L]}$ is generated from \eqref{eq:kernel-estimator} with $u_{ini} = u_{[t - \eta + 1, t]}$, $y_{ini} = y_{[t - \eta + 1, t]}$, and $u_{[1:mL]} = \bar{u}_{[0,L - 1]|t + 1}$.
Let $g_{t + 1} := (\bar{\textbf{K}} + \gamma I)^{-1} k(u_{[t - \eta + 1,t]}, y_{[t - \eta + 1, t]}, \bar{u}_{[0, L -1]|t + 1})$.
It can be observed from \eqref{eq:kernel-error-xi} that $(\bar{u}_{[0,L - 1]|t + 1},\bar{u}_{[0,L - 1]|t + 1},g_{t+1},h_{t+1})$ is a candidate solution for Problem \eqref{eq:ddkmpc-ideal}.

\emph{Step III.}
In the following, we derive a bound on the value function $J_L^*(\xi_{t + 1})$ using the candidate solution defined in the Step II above.
Building on the construction of $(\bar{u}_{[0,L - 1]|t + 1}$, $\bar{u}_{[0,L - 1]|t + 1}$, $g_{t+1}$, $h_{t+1})$, we have that
\begin{align}
      &J^*_L(\xi_{t + 1})  \nonumber\\
      &\le \sum_{i = 0}^{L - 1} \|\bar{u}_{i|t + 1}\|^2_{R} + \|{y}_{t + i + 1}\|^2_Q  + \lambda_h/\bar{d} \|h_{t + 1}\|^2 + \lambda_g \bar{d} \|g_{t + 1}\|^2 \nonumber\\
   &\le  -  \ell^*_{0|t} + \sum_{i = 0}^{i_\xi - 2} \ell^*_{i|t}  + \|y_{t + i+ 1}\|^2_Q - \| \bar{y}^*_{i+1|t} \|^2_Q  + \lambda_g \bar{d} \|g_{t + 1}\|^2 \nonumber\\
   &~~~ + \sum_{i = i_{\xi}-1}^{L - 1} \|\bar{u}_{i|t + 1}\|^2_R + \|y_{t + i + 1}\|^2_Q + \lambda_h/\bar{d} \|h_{t + 1}\|^2  \nonumber\\
   & \le -  \ell^*_{0|t} + J^*_{L}(\xi_t) + \underbrace{\sum_{i = 0}^{i_{\xi}-2}\|y_{t + i+ 1}\|^2_Q - \| \bar{y}^*_{i+1|t} \|^2_Q }_{\triangleq \textbf{I}_1} \nonumber\\
   &~~~ + \underbrace{ \sum_{i = i_{\xi}-1}^{L - 1} \|\bar{u}_{i|t + 1}\|^2_R + \|y_{t + i + 1}\|^2_Q}_{\triangleq \textbf{I}_2} \nonumber\\
   &~~~ + \underbrace{\lambda_h/\bar{d} \|h_{t + 1}\|^2 + \lambda_g \bar{d} \|g_{t + 1}\|^2}_{\triangleq \textbf{I}_3}
   \label{eq:next-J-decrease-1}.
\end{align}

Next, we derive upper bounds for $\mathbf I_1$, $\mathbf I_2$, and $\mathbf I_3$ separately.
Let $e_i:=y_{t+i+1}-\bar y^*_{i+1|t}$.
For $\mathbf I_1$, it follows that
\begin{align*}
\mathbf I_1
&=
\sum_{i=0}^{i_\xi-2}
\Big(
\|y_{t+i+1}\|_Q^2-\|\bar y^*_{i+1|t}\|_Q^2
\Big)\\
&=
\sum_{i=0}^{i_\xi-2}
\Big(
\|e_i\|_Q^2
+
2e_i^\top Q\bar y^*_{i+1|t}
\Big)\\
&\le
\sum_{i=0}^{i_\xi-2}\|e_i\|_Q^2
+
2
\Big(\sum_{i=0}^{i_\xi-2}\|e_i\|_Q^2\Big)^{1/2}
\Big(\sum_{i=0}^{i_\xi-2}\|\bar y^*_{i+1|t}\|_Q^2\Big)^{1/2}.
\end{align*}
Moreover, since $e_i
=
y_{t+i+1}
-
(\bar y^*_{i+1|t}+h^*_{i+1|t})
+
h^*_{i+1|t}$,
the kernel estimation error bound and $\sum_{i=0}^{L-1}\|h^*_{i|t}\|^2
\le
\frac{\bar d}{\lambda_h}J_L^*(\xi_t)
\le
\frac{\bar Y\bar d}{\lambda_h}$
imply $\sum_{i=0}^{i_\xi-2}\|e_i\|_Q^2
\le
2\lambda_{\max}(Q)
(
(i_\xi-1)p\bar d^2
+
\frac{\bar Y\bar d}{\lambda_h}
)$.
Since $\sum_{i=0}^{i_\xi-2}\|\bar y^*_{i+1|t}\|_Q^2
\le
\sum_{i=0}^{L-1}\|\bar y^*_{i|t}\|_Q^2
\le
Y_L(\xi_t)
\le
\bar Y$, we obtain $\mathbf I_1
\le
E_1(\bar d)+2\sqrt{\bar Y E_1(\bar d)}
=:\alpha_2(\bar d) \in\mathcal K_\infty$, where $E_1(\bar d):=
2\lambda_{\max}(Q)
(
(i_\xi-1)p\bar d^2
+
\frac{\bar Y\bar d}{\lambda_h}
)$.

In addition, 
it follows from \eqref{eq:f_ker+d} and \eqref{eq:Vs-property:decrease} that there exist a constant $c_0 > 0$ and a function $\hat{\alpha}_s(\cdot) \in \mathcal{K}_{\infty}$ such that
\begin{align}\label{eq:local-xi-decrease}
    \|\bar{\xi}_{i|t}\|^2 \le c_0 \rho_s^i \|\xi_t\|^2.
\end{align}
Therefore,
\begin{align*}
    \textbf{I}_2 &\le \lambda_{\max}(R,Q) \sum_{i = i_{\xi} - 1}^{L - 1} \|\xi_{t + i + 1}\|^2\\
    & \overset{\eqref{eq:local-xi-decrease}}{\le} \frac{\lambda_{\max}(R,Q) c_0}{1 - \rho_s} \|\xi_{t + i_\xi}\|^2 + (L - i_{\xi})\alpha_s(\bar{d})\\
    & \overset{\eqref{eq:xi-mid}}{\le } \frac{\lambda_{\max}(R,Q) c_0 \gamma_{\bar{Y}}}{\epsilon_o (L - 1)(1 - \rho_s)} \|\xi_t\|^2 \\
    &~~~+ \underbrace{\frac{\lambda_{\max}(R,Q) c_0 }{\epsilon_o (L - 1)(1 - \rho_s)}(\alpha_1(\bar{d}) + \hat{\alpha}_s(\bar{d})) + (L - i_{\xi})\alpha_s(\bar{d})}_{\triangleq \alpha_3(\bar{d})}.
\end{align*}
Furthermore, similar to the proof of Lemma \ref{lem:local-cost-bound}, we have that $\textbf{I}_3 {\le} (\lambda_h L p^2 + \lambda_g c_g )\bar{d}$.
Therefore, \eqref{eq:next-J-decrease-1} becomes
\begin{align}\label{eq:next-J-decrease}
    J^*_L(\xi_{t + 1}) &\le -\ell^*_{0|t} + J^*_{L}(\xi_t) + \frac{\lambda_{\max}(R,Q) c_0 \gamma_{\bar{Y}}}{\epsilon_o (L - 1)(1 - \rho_s)} \|\xi_t\|^2 \nonumber\\
    &~~~+ \underbrace{\alpha_2(\bar{d}) + \alpha_3(\bar{d}) + (\lambda_h L p^2 + \lambda_g c_g )\bar{d}}_{\triangleq \alpha_4(\bar{d})}.
\end{align}

On the other hand, note from \eqref{eq:Vo-property:decrease} that 
\begin{align*}
    &V_o(\xi_{t + 1}) - V_o(\xi_t) \\
    &\le - \epsilon_o \|\xi_t\|^2 + \|\bar{u}^*_{0|t}\|^2_R + \frac{1}{2}\|{y}_{t} - \bar{y}^*_{0|t} + \bar{y}^*_{0|t}\|^2_Q  \\
    &\le - \epsilon_o \|\xi_t\|^2 + \ell^*_{0|t} + \lambda_{\max}(Q) p^2 \bar{d}^2.
\end{align*}
Combining this inequality with \eqref{eq:next-J-decrease} yields
\begin{align*}
    &Y_{L}(\xi_{t + 1}) \\
    &\le J_L(\bar{\xi}_{t + 1}) \!+ V_o(\xi_t) - \epsilon_{o}\|\xi_t\|^2 \!+ \ell^*_{0|t} + \lambda_{\max}(Q)p^2 \bar{d}^2\\
    &\overset{\eqref{eq:next-J-decrease}}{\le}J^*_L(\xi_t) \!+ V_o(\xi_t) \!- \epsilon_o\underbrace{\Big(1 \!-\! \frac{\lambda_{\max}(R,Q) c_0 \gamma_{\bar{Y}}}{\epsilon_o^2 (L \!-\! 1)(1 \!-\! \rho_s)}\Big)}_{\triangleq a_L} \|\xi_t\|^2 \\
    &~~~+ \underbrace{\alpha_4(\bar{d}) + \lambda_{\max}(Q) p^2 \bar{d}^2}_{\triangleq \alpha_Y(\bar{d})} \\
    & \le Y_{L}(\xi_t) - a_L \epsilon_o\|\xi_t\|^2 + \alpha_Y(\bar{d}).
\end{align*}
To ensure that $Y_L$ is a practical Lyapunov function, we need $a_L >0$, which holds for a sufficiently long horizon $L$, i.e., $L > L_1:= 1 + \frac{\lambda_{\max}(R,Q) c_0 \gamma_{\bar{Y}}}{\epsilon_o^2 (1 - \rho_s)}$.
Moreover, to show recursive feasibility, we need that $Y_L(\xi_t) \le \bar{Y}$ holds for all $t >0$.
Based on \eqref{eq:YL-property:decrease}, there exists a small enough constant $\bar{d}_0 \le \min\{ \alpha_Y^{-1}(\bar{Y}a_L) , (\alpha_1 + \hat{\alpha}_s)^{-1}(\theta)\}$ such that $Y_L \le \bar{Y}$ holds for all $\bar{d} \le \bar{d}_0$.
In summary, the prediction horizon length $L$ must be such that 
\begin{equation}
    L > \max\{L_0,L_1\} := L_{\bar{Y}}
\end{equation}
and the interpolation error bound needs to satisfy $\bar{d} \le \bar{d}_0$.
This completes the proof.
\end{proof}

\section{Discussions
}
\label{sec:discussion}
Previous sections established the DDKPC framework under noise-free data and provided closed-loop stability guarantees. 
However, practical online implementation of Algorithm~\ref{alg:ddkmpc-ideal} faces two major challenges: 
i) the nonlinear and nonconvex structure of the DDKPC problem (cf.~\eqref{eq:ddkmpc-ideal:ker}) results in significant computational burden, and 
ii) measurement noise is unavoidable during both offline data collection and online operation. 
In the following, we address these issues separately.

\subsection{Computational feasibility via relaxation}
\label{sec:discussion-computation}

The optimization problem~\eqref{eq:ddkmpc-ideal} is inherently nonconvex due to the implicit kernel-based equality constraint~\eqref{eq:ddkmpc-ideal:ker}. 
Direct application of general-purpose nonlinear programming (NLP) solvers may therefore lead to prohibitive computational cost, limiting real-time applicability.

To enhance computational efficiency, the hard constraint in~\eqref{eq:ddkmpc-ideal:ker} can be incorporated into the objective function using quadratic penalty terms, as adopted in~\cite{huang2023robust}. 
Specifically, by introducing a penalized objective similar to~\cite[eq.~(29)]{huang2023robust}, Problem~\eqref{eq:ddkmpc-ideal} can be reformulated as $\min_{\overline{u} \in \mathcal{U},\, g_t} 
J_{\mathrm{relaxed}}(\overline{u}, g_t)
=
\sum_{i=0}^{L-1} 
\ell(\overline{u}_{i|t}, Y_F g_t)
+ \lambda_g \|g_t\|^2
+ \lambda_k 
\|(\bar{\mathbf K}+\gamma I)g_t 
- \mathbf k(u_{ini}, y_{ini}, \overline{u}_{[0,L-1]|t})\|^2$,
where $\lambda_k > 0$ is the penalty parameter for the kernel constraint violation.
The resulting problem admits efficient numerical treatment using gradient-based or first-order optimization methods, such as those developed in~\cite[Section~III]{wang2025optimization}, thereby reducing the computational burden compared with general NLP solvers.

It is worth noting that due to the nonconvex nature of the penalized objective (i.e., $\|(\bar{\mathbf K}+\gamma I)g_t 
- \mathbf k(u_{ini}, y_{ini}, \overline{u}_{[0,L-1]|t})\|^2$), convergence to a global minimizer cannot be guaranteed in general. 
Specifically, one must assume that the numerical solver is capable of
finding a suboptimal solution $(\bar u,\bar g_t)$ such that $J_{\rm relaxed}(\bar u,\bar g_t)-J_{\rm relaxed}^*(\xi_t)\le \delta_J$.
Since the implicit kernel constraint is enforced through a penalty term,
we further require the associated residual to be sufficiently small, i.e., $\|(\bar{\mathbf K}+\gamma I)\bar g_t
-\mathbf k(u_{ini},y_{ini},\bar u_{[0,L-1]|t})\|^2
\le \delta_K$.
Here $\delta_J,\delta_K>0$ are problem-dependent tolerances determined by
the available Lyapunov-decrease margin. The objective suboptimality
$\delta_J$ and the residual tolerance $\delta_K$ appear as additional
additive terms in \eqref{eq:YL-property:decrease} and can be absorbed into
the resulting practical bound by choosing them sufficiently small.
Such requirements are practically attainable by choosing a sufficiently
large penalty weight $\lambda_k$ and a sufficiently accurate solver
tolerance.
For instance, the methods proposed in
\cite{wang2025optimization} have demonstrated fast convergence, reaching a
numerical precision of less than $10^{-5}$ in a few iterations.

\subsection{Robust DDKPC under noisy data}
\label{sec:discussion-noise}

While Section \ref{sec:ddkmp:noise-free} establishes theoretical guarantees under noise-free conditions, practical implementations must account for measurement noise during both offline data collection and online operation. Suppose the exact output $y_t = h_0(x_t)$ is corrupted by additive noise $n_t$, yielding the available measurements $\tilde{y}_t = h_0(x_t) + n_t$.

\begin{assumption} \label{as:noise}
    The measurement noise $n_t$ is uniformly bounded such that $\|n_t\| \le \bar{n}$ for all $t \in \mathbb{N}$, where $\bar{n} > 0$ is a known constant.
\end{assumption}

To maintain robustness, the offline data matrices are constructed using the noise-corrupted data, denoted as $\tilde{\textbf{K}}$ and $\tilde{Y}_F$. Furthermore, we aggregate the kernel approximation error and the measurement noise into a unified uncertainty bound $\bar{w} := \bar{n} + \bar{d}$. To compensate for the initial condition mismatch induced by $\tilde{y}_t$, the nominal formulation \eqref{eq:ddkmpc-ideal} is robustified as:
\begin{subequations}\label{eq:ddkmpc-noise}
    \begin{align}
       J_L(\tilde{\xi}_t) := \label{eq:ddkmpc-noise:obj}\\
       \min_{g,\bar{u},\bar{y},h} ~~&\sum_{i = 0}^{L - 1} \ell(\bar{u}_{i|t}, \bar{y}_{i|t}) + \frac{\lambda_h}{\bar{w}}\|h_t\|^2 + \lambda_g \bar{w}\|g_t\|^2 \nonumber\\
         {\rm s.t.}
~~& \bmat{
       \tilde{\textbf{K}} + \gamma I\\
        \tilde{Y}_F
        } g_t = \bmat{
       \textbf{k}(u_{ini},y_{ini},\bar{u}_{[0,L - 1]|t})\\
        \bar{y}_{[0,L - 1]|t} + h_{[0,L - 1]|t}
        },
\label{eq:ddkmpc-noise:ker}\\
        &\bmat{
        u_{ini}\\
        y_{ini} + h_{[-\eta, -1]|t}
        } = \bmat{
        u_{[t - \eta,t - 1]}\\
        \tilde{y}_{[t - \eta,t - 1]}
       },\label{eq:ddkmpc-noise:ini}\\
        & \bar{u}_{i|t} \in \mathcal{U},~~\forall i \in [0,L - 1],\label{eq:ddkmpc-noise:h-u-constraint}
    \end{align}  
\end{subequations}
where the slack variable vector is augmented as $h_t := \smat{h_{[-\eta, -1]|t}\\ h_{[0, L - 1]|t}}$ to penalize both the prediction mismatch and the initial condition discrepancy.

Compared to the nominal formulation \eqref{eq:ddkmpc-ideal}, Problem \eqref{eq:ddkmpc-noise} rescales the regularization weights via the unified bound $\bar{w}$ to absorb the measurement noise. Since the kernel representation error bound \eqref{eq:kernel-error} also holds for noise-corrupted data under regularized interpolation (cf. \cite{lahr2025kernelbound}), the closed-loop stability guarantees can be extended to this robustified framework.

\begin{theorem}[Robust stability under noisy data]
\label{thm:stab-noise}
    Let Assumptions \ref{as:io-equiv}--\ref{as:noise} hold. 
    Then, for any constant $\tilde{Y} >0$, there exist constants $L_{\tilde{Y}}, \gamma_{\tilde{Y}}, \bar{w}_0 >0$ such that for $Y_L(\xi_0) \le \tilde{Y}$, $L > L_{\tilde{Y}}$, and $\bar{w} \le \bar{w}_0$, Problem \eqref{eq:ddkmpc-noise} is feasible for all times $t >0$. Furthermore, the function $Y_L$ defined in \eqref{eq:Lyp-func} fulfills:
   \begin{subequations}\label{eq:YL-property-noise}
        \begin{align}
            \epsilon_o \|\tilde{\xi}_t\|^2 \le Y_{L}(\tilde{\xi}_t) &\le \gamma_{\tilde{Y}} \|\tilde{\xi}_t\|^2 + \hat{\alpha}_s(\bar{w}) \label{eq:YL-property-noise:l-u-bound}\\
           Y_{L}(\tilde{\xi}_{t + 1}) - Y_{L}(\tilde{\xi}_t) &\le -a_L \epsilon_o \|\tilde{\xi}_t\|^2 + \alpha_Y(\bar{w}). \label{eq:YL-property-noise:decrease}
 \end{align}
    \end{subequations}
\end{theorem}

\begin{proof}
    The proof follows a similar step as established in Theorem \ref{thm:stab-ideal}. The key modification lies in accommodating the initialization mismatch \eqref{eq:ddkmpc-noise:ini}. Given a feasible solution at time $t$, we construct the candidate solution at time $t+1$ identically to Step II of Theorem \ref{thm:stab-ideal}, with the additional assignment of the initial slack candidate as $h_{[-\eta, -1]|t+1} = -n_{[t - \eta + 1, t]}$. 
    
    Under Assumption \ref{as:noise}, the norm of this candidate slack is strictly bounded by $\sqrt{\eta}\bar{n} \le \sqrt{\eta}\bar{w}$. Consequently, the penalty term $\frac{\lambda_h}{\bar{w}}\|h_t\|^2$ remains proportionately bounded by $\bar{w}$. By aggregating the measurement noise and the approximation error into the unified uncertainty bound $\bar{w}$, the upper bound derivations for the candidate cost scale symmetrically to those in Lemma \ref{lem:local-cost-bound}. Substituting $\bar{d}$ with $\bar{w}$ throughout Step III of Theorem \ref{thm:stab-ideal} directly yields the practical Lyapunov decrease inequality \eqref{eq:YL-property-noise:decrease}, completing the proof.
\end{proof}

\section{Online DDKPC for Slowly Time-Varying Nonlinear Systems}
\label{sec:online-ddkpc}
It can be observed from Algorithm \ref{alg:ddkmpc-ideal} that the dataset $\mathcal{D}$ used to construct the matrices $Y_F$ and $\bar{\mathbf{K}}$ is fixed offline and is not updated during closed-loop operation. 
Such a static design is appropriate when the effective multi-step input-output mapping seen by the controller remains approximately stationary over the operating regime of interest. 
In many practical applications, however, this mapping may drift over time due to e.g., changing operating conditions, actuator aging, payload variation, tire-road changes, or other unmodeled effects. 
As a result, a predictor constructed from a fixed offline dictionary may gradually lose relevance along the evolving closed-loop trajectory. 
This motivates the development of an online extension of the proposed DDKPC framework followed by closed-loop stability guarantees. 

For simplicity, we focus only on the noise-free case. 
The extension to noisy online measurements follows the same robustification principle as in Section \ref{sec:discussion-noise}.

\subsection{Time-varying nonlinear systems}

Consider a time-varying nonlinear system as follows
\begin{subequations}
\label{eq:sys-tv}
\begin{align}
x_{t+1} &= f_t^0(x_t,u_t), \label{eq:tv-sys:x}\\
y_t &= b_t^0(x_t), \label{eq:tv-sys:y}
\end{align}
\end{subequations}
where functions $f_t^0(\cdot,\cdot)$ and $b_t^0(\cdot)$ are unknown and may vary with time.
Similar to Assumption \ref{as:io-equiv}, the following assumption is imposed to characterize the input-output behavior of \eqref{eq:sys-tv}.

\begin{assumption}
\label{as:tv-io}
There exist a known constant $\eta \in \mathbb{N}_{+}$, sets $\bar{\Xi} \subset \mathbb{R}^{(m+p)\eta}$ and $\bar{\mathcal{U}}\subset\mathbb{R}^{m}$ such that for all $(\xi_t,u_t)\in \bar{\Xi} \times \bar{\mathcal{U}}$, the extended state $\xi_t$ defined in \eqref{eq:xi} satisfies
\begin{subequations}
\label{eq:sys-tv-xi}
\begin{align}
\xi_{t+1} &= f_t(\xi_t,u_t), \label{eq:sys-tv-xi:xi}\\
y_t &= b_t(\xi_t), \label{eq:sys-tv-xi:y}
\end{align}
\end{subequations}
and reproduces the same input-output behavior as \eqref{eq:sys-tv}.
\end{assumption}

Building on Assumption \ref{as:tv-io}, the time-varying version of the multi-step predictor in \eqref{eq:multi-step-predictor} is given by
\begin{equation}
\label{eq:tv-predictor}
\bar{y}_{t,[a]}(\zeta) = q_{t,[a]}(\zeta).
\end{equation}

In the following, we focus our attention on a class of slow time-varying system, i.e., we assume that functions $f_t(\cdot,\cdot)$ and $b_t(\cdot)$ have a bounded rate of variation, which is formally imposed in the following assumption.

\begin{assumption}
\label{as:tv-bound}
There exist $L_f,L_b\ge 0$ such that for all $t_1,t_2\in\mathbb{N}$, $\xi\in\bar{\Xi}$, $u_{t_2+i_0}\in\bar{\mathcal{U}}$ with $i_0 \in [0,i - 1]$, and $i\in[0,L-1]$, the following inequalities hold:
\begin{subequations} 
\label{eq:tv-bound}
\begin{align}
&\left\|
\Phi_{t_1,i}^{t_1}\big(\xi,u_{[t_2,t_2+i-1]}\big)
-
\Phi_{t_2,i}\big(\xi,u_{[t_2,t_2+i-1]}\big)
\right\|
\nonumber\\
&\hspace{3em}\le
L_f i\big(|t_1-t_2|+i\big)\|\xi\|,
\label{eq:tv-bound:f}\\
&\left\|
b_{t_1}\circ
\Phi_{t_1,i}^{t_1}\big(\xi,u_{[t_2,t_2+i-1]}\big)
-
b_{t_2+i}\circ
\Phi_{t_2,i}\big(\xi,u_{[t_2,t_2+i-1]}\big)
\right\|
\nonumber\\
&\hspace{3em}\le
L_b (i+1)\big(|t_1-t_2|+i\big)\|\xi\|.
\label{eq:tv-bound:b}
\end{align}
\end{subequations}
Here, $\Phi_{t_1,0}^{t_1}(\xi,\cdot)=\Phi_{t_2,0}(\xi,\cdot)=\xi$, and for $i\ge 1$, $\Phi_{t_1,i}^{t_1}\big(\xi,u_{[t_2,t_2+i-1]}\big)
:=
f_{t_1}^{[i]}\big(\xi,u_{[t_2,t_2+i-1]}\big)$, $\Phi_{t_2,i}\big(\xi,u_{[t_2,t_2+i-1]}\big)
:=
f_{t_2+i-1}\circ f_{t_2+i-2}\circ\cdots\circ f_{t_2}
\big(\xi,u_{[t_2,t_2+i-1]}\big)$.
\end{assumption}

\begin{remark}
Assumption~\ref{as:tv-bound} imposes a bounded-rate condition directly
on the multi-step input-output evolution maps. It can be regarded as a
nonlinear counterpart of the Lipschitz regularity commonly imposed on
linear time-varying (LTV) matrix trajectories, e.g.,
$\|[A(t)-A(s)\;\;B(t)-B(s)]\|\le L|t-s|$ in
\cite[Assumption~1]{liu2023timevarying}, which similarly controls the
deviation between frozen and time-varying evolutions over a finite
horizon. This viewpoint is also consistent with averaging-based analyses
of nonlinear time-varying systems, where the behavior of the actual
time-varying system is related to that of an associated averaged or
frozen system; see, e.g., \cite{nesic2001input}.
\end{remark}

To capture the time-varying behavior of \eqref{eq:tv-predictor}, parameters $Y_F$, $\bar{\textbf{K}}$ and $\textbf{k}$ in Algorithm \ref{alg:ddkmpc-ideal} should be updated online.
To this end,  the following subsections are conducted by first introducing a data update mechanism, followed by a parameter updating rule and the online DDKPC.

\subsection{Online dictionary update via ALD}
Adapting standard online data-updating schemes from linear data-driven control to the kernel-based setting requires additional care. 
Unlike standard online DDPC \cite{berberich2021linearII}, where updating data matrices relies on a straightforward mechanism of appending and dropping columns, the extension to a kernel-based framework presents two primary challenges. 
First, incorporating new data necessitates the recalculation of the kernel Gram matrix $\bar{\mathbf{K}}$, which introduces a significant computational burden as the dataset grows. 
Second, a conventional First-In-First-Out (FIFO) replacement rule \cite{berberich2021linearII} may discard samples that remain informative for previously visited nonlinear operating regions, thereby degrading the global representational quality of the predictor. 
To address these issues while preserving the mathematical consistency and stability guarantees of the overall framework, we adopt a fixed-budget online dictionary management strategy based on the ALD criterion \cite{xu2007kernel}.

Let $\mathcal{D}_t^{\rm a}
= \{(\zeta_1,Y_{F,1}),\dots,(\zeta_{N_t},Y_{F,N_t})\}$
denote the active data dictionary at time $t$, with a bounded capacity
$N_t \le M_{\rm dict}$. During online operation, a candidate data point is
formed only after the corresponding length-$(\eta+L)$ input-output window
becomes available. Specifically, after the pair $(u_t,y_t)$ has been
recorded, the completed window anchored at $\tau=t-L+1$ is used to
construct $\zeta_{ new}
:=\mathrm{col}\big(u_{[\tau-\eta,\tau-1]},
y_{[\tau-\eta,\tau-1]},
u_{[\tau,\tau+L-1]}\big)$ and $
\qquad
Y_{F,{ new}}:=y_{[\tau,\tau+L-1]}$.
Then, its informational novelty is evaluated under the regularized model
by calculating the ALD criterion
\begin{equation} \label{eq:ald_condition}
\delta_t =
K(\zeta_{\rm new}, \zeta_{\rm new})
-\mathbf{k}_t(\zeta_{\rm new})^\top
({\bf{\bar{K}}}_t + \gamma I)^{-1}
\mathbf{k}_t(\zeta_{\rm new})
\end{equation}
where $\mathbf{k}_t(\zeta_{new}) = [K(\zeta_1^d, \zeta_{new}), \dots, K(\zeta_{N_t}^d, \zeta_{new})]^\top$. Specifically, let $\nu > 0$ be a pre-defined novelty threshold. If $\delta_t \le \nu$, the new data provides negligible information gain relative to the existing dictionary and the regularization level, and is thus discarded ($\mathcal{D}_{t+1}^a = \mathcal{D}_t^a$). 
Conversely, if $\delta_t > \nu$, $\zeta_{new}$ contains unmodeled dynamics and must be admitted. 
If the dictionary is unsaturated ($N_t < M_{\rm dict}$), $\zeta_{new}$ is simply appended. 
However, if the dictionary operates at its maximum budget ($N_t = M_{\rm dict}$), a pruning operation is triggered. 
To prevent severe deterioration of the prediction accuracy, we do not arbitrarily discard the oldest data, but seek to remove the most redundant element $\zeta_j \in \mathcal{D}_t^a$. 
This is achieved by identifying the element with the
smallest regularized conditional contribution to the span of the remaining dictionary elements:
\begin{equation}\label{eq:pruning_rule}
\iota=
\arg\max_{i\in[1,N_t]}
\left[(\bar{\mathbf K}_t+\gamma I)^{-1}\right]_{i,i}.
\end{equation}
Here, the $\iota$th element is removed, and $\zeta_{new}$ is inserted to form $\mathcal{D}_{t+1}^a$.

\begin{remark}
\label{rem:recency}
In practical implementations, the online dictionary update mechanism in \eqref{eq:ald_condition} and \eqref{eq:pruning_rule} can be augmented with redundancy safeguards by e.g., retaining a small quota of recent samples, imposing a bounded maximum sample age, or avoiding the immediate removal of newly admitted informative points.
Such engineering heuristics are not required for the theoretical development below, but may improve robustness of the online update mechanism in practice.
\end{remark}

\subsection{Online DDKPC and closed-loop guarantees}

Using the current dataset $\mathcal{D}_t^a$, the most direct way to formulate an online version of the kernel predictor in \eqref{eq:kernel-estimator} is to let $y^{\rm ker}_{t,[1:pL]}
=
Y_{F,t}(\bar{\mathbf{K}}_t + \gamma I)^{-1}\mathbf{k}_t(\zeta)$,
where $\bar{\mathbf{K}}_t$, $Y_{F,t}$ and $\mathbf{k}_t(\zeta)$ are constructed from the active dictionary $\mathcal{D}_t$.
Although updating these quantities at every sampling instant may improve the local prediction accuracy, it also changes the implicit predictor entering the DDKPC problem at every step, which may introduce undesirable fluctuations in the closed-loop optimization. 
To balance predictor adaptation and closed-loop consistency, we separate the online maintenance of the active dictionary from the periodic reconstruction of the predictor used in the DDKPC problem.
The active dictionary $\mathcal{D}_t^{\rm a}$ is maintained online by the ALD criterion: whenever a completed length-$(\eta+L)$ input-output window becomes available, it is tested by \eqref{eq:ald_condition} and either discarded or inserted into $\mathcal{D}_t^{\rm a}$ according to the fixed-budget rule in \eqref{eq:pruning_rule}. In contrast, the predictor used in the online DDKPC problem is reconstructed only periodically from a frozen snapshot of the active dictionary.

Specifically, given an updating period $T_0 \in \mathbb{N}_+$.
For $j \in \mathbb{N}$, let $s_j = j T_0$ denote the $j$th controller update time instant.
At $t = s_j$, $j \in \mathbb{N}$, freeze the current active dictionary $\mathcal{D}_{s_j}^c = \mathcal{D}_t^a$ and compute $\bar{\textbf{K}}_{s_j}$, $Y_{F,s_j}$, and $\textbf{k}_{s_j}$.
Given $\hat{d} > 0$, for any $t \in [s_j, s_{j + 1} - 1]$, the online DDKPC solves
\begin{subequations}
\label{eq:ddkmpc-online}
\begin{align}
\hat{J}_{L}&(\xi_t) :=  \label{eq:ddkmpc-online-obj}\\
& \min_{g_t,\bar{u},\bar{y},h_t}
\sum_{i=0}^{L-1}\ell(\bar{u}_{i|t},\bar{y}_{i|t})
+\frac{\lambda_h}{\hat{d}}\|h_t\|^2
+\lambda_g \hat{d}\|g_t\|^2 \nonumber
\\
{\rm s.t.}~~&
\bmat{\!
\bar{\mathbf{K}}_{s_j}\!+\!\gamma I\\
Y_{F,{s_j}}
}g_t
\!\!=\!\!
\bmat{
\mathbf{k}_{s_j}(u_{\rm ini},y_{\rm ini},\bar{u}_{[0,L-1]|t})\\
\bar{y}_{[0,L-1]|t}+h_{[0,L-1]|t}
},
\label{eq:ddkmpc-online-ker}\\
&
\bmat{
u_{\rm ini}\\
y_{\rm ini}
}
=
\bmat{
u_{[t-\eta,t-1]}\\
y_{[t-\eta,t-1]}
},
\label{eq:ddkmpc-online-ini}\\
&
\bar{u}_{i|t}\in\mathcal{U},\qquad \forall i\in[0,L-1].
\label{eq:ddkmpc-online-u}
\end{align}
\end{subequations}
Here, constraint \eqref{eq:ddkmpc-online-ker} provides an implicit form of the predictor
\begin{equation} 
\label{eq:kernel-estimator-tv}
    y^{\rm ker}_{t,[1:pL]}
=
Y_{F,s_{j}}(\bar{\mathbf{K}}_{s_{j}} + \gamma I)^{-1}\mathbf{k}_{s_{j}}(\zeta).
\end{equation}
To ensure the accuracy of the kernel-based estimator in the online setting, we modify Assumption \ref{as:bar_d} as follows.
\begin{assumption}
\label{as:est-error-bound-tv}
There exists a constant $\bar d_{on}\ge 0$ such that, for all $s_j = jT_0$, $j \in \mathbb{N}$, all $\zeta \in \bar{\Omega}$, and all $a\in[1,pL]$,
\begin{equation}
\label{eq:kernel_est_bound_on}
|y^{ker}_{s_{j}^{},[a]}(\zeta) - q_{s_{j}^{},[a]}(\zeta)|
\le d_{s_{j}^{}}(\zeta)\le \bar d_{ on}.
\end{equation}
\end{assumption}
Combining Assumptions \ref{as:tv-bound} and \ref{as:est-error-bound-tv}, we choose $\hat{d} \ge \bar{d}_{on} + L_f + L_b$.
The corresponding online execution procedure is summarized in Algorithm \ref{alg:ddkmpc-online}.

\begin{algorithm}[h]
    \caption{Online DDKPC.}
    \label{alg:ddkmpc-online}
    \begin{algorithmic}[1]
        \STATE \textbf{Input:} Initial active dictionary $\mathcal{D}^{\rm a}_0$, budget $M_{\rm dict}$, novelty threshold $\nu$, update period $T_0$, feasible set $\mathcal{U}$, prediction horizon $L$, matrices $Q$, $R$, and parameters $\lambda_h$, $\lambda_g$, $\hat d$, $\gamma$.
        \STATE \textbf{Initialize:} Given $s_0=0$, set $\mathcal{D}^{\rm c}_{s_0}=\mathcal{D}^{\rm a}_0$.
        \STATE Construct $\bar{\mathbf K}_{s_0}$ and $Y_{F,s_0}$ from $\mathcal{D}^{\rm c}_{s_0}$.
        \FOR{$t=0,1,2,\ldots$}
            \STATE Set $j=\lfloor t/T_0\rfloor$ and $s_j=jT_0$.
            \IF{$t=s_j$ and $t>0$}
                \STATE Freeze the current active dictionary $\mathcal{D}^{\rm c}_{s_j}=\mathcal{D}^{\rm a}_{t}$.
                \STATE Construct $\bar{\mathbf K}_{s_j}$ and $Y_{F,s_j}$ from $\mathcal{D}^{\rm c}_{s_j}$.
            \ENDIF
            \STATE Solve Problem \eqref{eq:ddkmpc-online} using $\bar{\mathbf K}_{s_j}$, $Y_{F,s_j}$, $\mathbf k_{s_j}(\cdot)$, and the most recent $\eta$ measurements $(u_{[t-\eta,t-1]},y_{[t-\eta,t-1]})$.
            \STATE Apply the input $u_t=\bar u^{*}_{0|t}$ to system \eqref{eq:sys-tv}.
            \STATE Record the new input-output pair and update the data buffer.
            \IF{a completed length-$(\eta+L)$ window is available}
                \STATE Form $(\zeta_{{ new}},Y_{F,{ new}})$.
                \STATE Compute $\delta_t$ from \eqref{eq:ald_condition} using $\mathcal{D}^{\rm a}_{t}$.
                \IF{$\delta_t > \nu$}
                    \IF{$|\mathcal{D}^{\rm a}_{t}| = M_{\rm dict}$}
                        \STATE Compute $\iota$ from \eqref{eq:pruning_rule} and remove $(\zeta_{\iota},Y_{F,\iota})$ from $\mathcal{D}^{\rm a}_{t}$.
                    \ENDIF
                    \STATE Insert $(\zeta_{{ new}},Y_{F,{ new}})$ into $\mathcal{D}^{\rm a}_{t}$ to obtain $\mathcal{D}^{\rm a}_{t+1}$.
                \ELSE   
                \STATE  Set $\mathcal{D}^{\rm a}_{t+1}=\mathcal{D}^{\rm a}_{t}$.
                \ENDIF
            \ELSE 
                \STATE Set $\mathcal{D}^{\rm a}_{t+1}=\mathcal{D}^{\rm a}_{t}$.
            \ENDIF
        \ENDFOR
    \end{algorithmic}
\end{algorithm}

The online optimization problem \eqref{eq:ddkmpc-online} is updated at
each instant $s_j$ through the reconstructed kernel predictor
$(\bar{\mathbf K}_{s_j},Y_{F,s_j},\mathbf k_{s_j})$. To analyze the
closed-loop behavior under these updates, we use the following uniform
versions of Assumptions~\ref{as:stab} and \ref{as:observ}.
\begin{assumption}[Uniform local stabilizability]
\label{as:stab-tv}
There exist a continuous Lyapunov function $V_{s,j}(\xi)$, a feedback law $\kappa_j(\xi)$, a compact set $\mathcal{I}$, constants $c_{s,l},c_{s,u},\delta^{loc}>0$, $\rho_s\in(0,1)$, and a function $\alpha_s\in\mathcal{K}_\infty$, all independent of $j$, such that, for all $s_j=jT_0$ with $\xi\in \mathbb{X}_{\delta,j}:=\{\xi\mid V_{s,j}(\xi)\le \delta^{loc}\}$ and $d\in\mathcal{I}$, it holds that $f_{s_j}(\xi,\kappa_j(\xi))+d\in\mathbb{X}_{\delta,j}$, and
\begin{subequations}
\label{eq:stab-tv}
\begin{align}
&c_{s,l}\|\xi\|^2
\le
V_{s,j}(\xi)
\le
c_{s,u}\|\xi\|^2,
\label{eq:stab-tv:bound}\\
&V_{s,j}\big(f_{s_j}(\xi,\kappa_j(\xi))+d\big)
\le
\rho_s V_{s,j}(\xi)+\alpha_s(\|d\|).
\label{eq:stab-tv:decrease}
\end{align}
\end{subequations}
\end{assumption}

\begin{assumption}[Common robust UIOSS]
\label{as:uioss-tv}
There exist a continuous function \(V_o\), constants
\(c_{o,l},c_{o,u},\epsilon_o>0\), a compact set $\mathcal{I}$ , and a function
\(\alpha_o\in\mathcal K_\infty\) such that, for all
\(s_j=jT_0\), \((\xi,u) \in \bar{\Xi} \times  \bar{\mathcal{U}}\), and \(d \in \mathcal{I}\), the following inequalities hold:
\begin{subequations}
\label{eq:uioss-tv}
    \begin{align}
        &c_{o,l}\|\xi\|^2\le V_o(\xi)\le c_{o,u}\|\xi\|^2, \label{eq:uioss-tv:bound}\\
        &V_o(f_{s_j}(\xi,u)+d)-V_o(\xi)
\le
-\epsilon_o\|\xi\|^2
+\frac{1}{2}\|b_t(\xi)\|_Q^2
+\|u\|_R^2 \label{eq:uioss-tv:decrease}.
    \end{align}
\end{subequations}
\end{assumption}

With these Assumptions in place, we first establish a local
feasibility and cost bound for the online problem.

\begin{lemma}
\label{lem:local-cost-bound-tv}
Under Assumptions \ref{as:tv-io}--\ref{as:uioss-tv}, Problem \eqref{eq:ddkmpc-online} is feasible for all $\|\xi_t\|^2\le \delta_s$, and there exist a constant $\hat{\gamma}_s>0$ and a function $\hat{\alpha}_s^{\rm on}(\cdot)\in\mathcal{K}_{\infty}$ such that
\begin{equation}
\label{eq:local-cost-bound-on-tv}
\hat{J}_{L}^{*}(\xi_t)\le \hat{\gamma}_s\|\xi_t\|^2+\hat{\alpha}_s^{on}(\hat d).
\end{equation}
\end{lemma}

\begin{proof}
For $t \in [s_j, s_{j + 1} - 1]$, consider a Lyapunov function $V_{s,j}(\xi)$ and its corresponding feedback $\kappa_{j}(\xi)$ satisfying Assumption \ref{as:stab-tv}.
In the following, we use $y_{j,t}$ to represent the output generated from the time-invariant nonlinear system at time $s_j$ as follows
\begin{subequations}\label{eq:sys-tv-sj}
    \begin{align}
        \xi_{j,t + 1} & = f_{s_j}(\xi_t, u_t)\\
        y_{j,t} & = b_{s_j}(\xi_t).
    \end{align}
\end{subequations}
For $\|\xi_t\| \le \delta_s$, Assumption \ref{as:stab-tv} implies that $V_{s,j}(\xi_t) \le \delta^{loc}$.
Hence, we can let $\bar{u}_{0|t} = \kappa_j(\xi_t)$.
Applying this feedback control input into \eqref{eq:sys-tv-xi} yields
\begin{align*}
    &\xi_{t + 1} = f_{t}(\xi_t, \kappa_j(\xi_t)) \\
    &=  f_{s^{}_j}(\xi_t, \kappa_j(\xi_t))  
    +  (\underbrace{ f_{t}(\xi_t, \kappa_j(\xi_t)) - f_{s^{}_j}(\xi_t, \kappa_j(\xi_t))}_{\triangleq d_j} 
\end{align*}
where $\|d_j\|  \le L_f |t - s_j| \|\xi_t\| \le L_f T_0 \delta_s$.
Hence, Assumption \ref{as:stab-tv} implies that for small enough $L_f$, it holds that $\|\xi_{t + 1}\| \le \delta_s$ and we can construct $\bar{u}_{1|t} = \kappa_j(\xi_{t + 1})$.
Recursively, for small enough $L_f$, we are able to construct our candidate input as $\bar{u}_{i|t} = \kappa_j(\xi_{i+t})$, $i \in [0, L - 1]$.
In addition, based on Assumption \ref{as:stab-tv}, for $j \in \mathbb{N}$, there exist constants $\bar{c}_{0} >0$, $\rho_s \in (0, 1)$, and a function $\bar{\alpha}_{s}(\cdot) \in \mathcal{K}_{\infty}$ such that the resulting $\xi_{t + i}$ from \eqref{eq:sys-tv-xi} satisfying
\begin{equation}
\label{eq:xi-bound-tv}
    \|\xi_{t + i}\|^2 \le \bar{c}_{0} \rho_{s}^i \|\xi_t\|^2 + \bar{\alpha}_{s}(L_f).
\end{equation}
Next, we construct candidate solutions $(\bar{y}_{[0, L - 1]|t}, g_t, h_t)$ for Problem \eqref{eq:ddkmpc-online}.
Let $\bar{y}_{i|t}$ be the resulting output from \eqref{eq:sys-tv-xi}, $g_t := (\bar{\mathbf{K}}_{s_j} + \gamma I)^{-1} \textbf{k}_{s_j}(\xi_t, \bar{u}_{[0,L - 1]|t})$, and $h_{[0,L - 1]|t} = y^{ker}_{[0,L - 1]|t} - y_{[t, t + L- 1]}$.
Here, $\|g_t\|^2 \le c_g$ according to Assumption \ref{as:bound-kernel} and 
\begin{align}
\label{eq:h-bound-tv}
    \|h_t\|^2 & \le 2 \|y^{ker}_{[0,L - 1]|t} - y_{j,[t, t + L - 1]}\|^2 \nonumber \\
    &~~~ + 2 \|y_{j,[t, t + L - 1]} - y_{[t,t + L - 1]}\|^2 \nonumber \\
    &\le  2 L p \bar{d}_{on}^2 + 2 (\sum_{i = 0}^{L - 1}L_b(T_0 + i)i)^2 \|\xi_t\|^2. 
\end{align}
Note that $\hat{d} \ge \bar{d}_{on} + L_f + L_b $.
It can be derived that 
\begin{align*}
    J^*(\xi_t) & \le \sum_{i = 0}^{L - 1} \|\bar{u}_{i|t}\|_R^2 + \|\bar{y}_{i|t}\|_Q^2 + \lambda_h/\hat{d}\|h_t\|^2 + \lambda_g \hat{d} \|g_t\|^2\\
    & \overset{\eqref{eq:h-bound-tv}}{\le} \lambda_{\max}(Q,R)\sum_{j = 1}^{L}\|\bar{\xi}_{i|t}\|^2 + 2 \lambda_h L p \hat{d} \\
    &~~+ 2(\lambda_h/\hat{d})(\sum_{i = 0}^{L - 1}L_b(T_0 + i)i)^2\|\xi_t\|^2 + \lambda_g c_g \hat{d} \\
    & \overset{\eqref{eq:xi-bound-tv}}{\le} \hat{\gamma}_{s} \|\xi_t \|^2 + \hat{\alpha}_{s}^{on}(\hat{d}) 
\end{align*}
where $\hat{\gamma}_s := \lambda_{\max}(Q,R)\sum_{j = 1}^{L} c_0 \rho_s^i + 2(\lambda_h/\hat{d})(\sum_{i = 0}^{L - 1}L_b(T_0 + i)i)^2$ and $\hat{\alpha}_{s}^{on}(\hat{d}) := \lambda_{\max}(Q,R) L \bar{\alpha}_s(\hat{d}) + 2 \lambda_h L p \hat{d} + \lambda_g c_g \hat{d}$.
This completes the proof.
\end{proof}

Building on this lemma, we are ready to show recursive feasibility and closed-loop stability.
Define the Lyapunov candidate by
\begin{equation}
\label{eq:online-lyapunov}
\hat{Y}_{L}(\xi_t):=\hat{J}_{L}^{*}(\xi_t)+V_{o}(\xi_t).
\end{equation}

\begin{theorem}[Practical stability of online DDKPC]
\label{thm:online-stability}
Suppose that Assumptions \ref{as:tv-io}--\ref{as:uioss-tv} hold.
Then, for any constant $\hat{Y}>0$, there exist constants
$L_{\hat{Y}}, \gamma_{\hat{Y}}>0$ and $\hat{d}_{0}>0$ such that, if
$\hat{Y}_{L}(\xi_0)\le \hat{Y}$, $L>L_{\hat{Y}}$, and $\bar{d}_{\rm on}+L_f+L_b\le \hat{d} \le \hat d_0$,  Problem \eqref{eq:ddkmpc-online} is feasible for all $t\ge 0$.
Moreover, the Lyapunov function candidate \eqref{eq:online-lyapunov} satisfies
\begin{subequations}
\label{eq:online-lyap-ineq}
\begin{align}
&\epsilon_o\|\xi_t\|^2
\le
\hat{Y}_{L}(\xi_t)
\le
\gamma_{\hat{Y}}\|\xi_t\|^2+\hat{\alpha}_s^{ on}(\hat d),
\label{eq:online-lyap-ineq:bound}\\
&\hat{Y}_{L}(\xi_{t+1})-\hat{Y}_{L}(\xi_t)
\le
-\hat{a}_L\epsilon_o\|\xi_t\|^2+\hat{\alpha}_{\hat{Y}}(\hat d),
\label{eq:online-lyap-ineq:decrease}
\end{align}
\end{subequations}
where $\hat{a}_L>0$ and $\hat{\alpha}_{\hat{Y}}(\cdot)\in\mathcal{K}_{\infty}$.
\end{theorem}

\begin{proof}
The proof follows the same structure as that of Theorem \ref{thm:stab-ideal}.
We only detail the additional estimates caused by the time-varying dynamics.

First, by Lemma \ref{lem:local-cost-bound-tv} and \eqref{eq:uioss-tv:bound}, there exists $\gamma_{\hat Y}>0$ such that the upper bound in
\eqref{eq:online-lyap-ineq:bound} holds on the sublevel set
$\{\xi:\hat Y_L(\xi)\le \hat Y\}$.
The lower bound follows from $\hat J_L^*(\xi_t)\ge 0$ and
\eqref{eq:uioss-tv:bound}; without loss of generality, $\epsilon_o$ is taken no larger than $c_{o,l}$.

Next, assume that Problem \eqref{eq:ddkmpc-online} is feasible at time $t$, and denote its optimal solution by $(\bar{u}^*_{[0,L-1]|t},\bar{y}^*_{[0,L-1]|t},g_t^*,h_t^*)$.
Let $j=\lfloor t/T_0\rfloor$ and $s_j=jT_0$.
Denote by $\xi_{[t+1,t+L]}$ and $y_{[t,t+L-1]}$ the trajectory of the actual time-varying system \eqref{eq:sys-tv-xi} generated by the input sequence $\bar{u}^*_{[0,L-1]|t}$.
For each $i\in[0,L-1]$, define the time-variation-induced perturbation by $d_{j,t+i}
:=
f_{t+i}(\xi_{t+i},\bar u^*_{i|t})
-
f_{s_j}(\xi_{t+i},\bar u^*_{i|t})$.
By Assumption \ref{as:tv-bound} and the choice of $\hat d_0$, this perturbation belongs to the admissible compact set $\mathcal I$ along the considered trajectory. Hence, using \eqref{eq:uioss-tv:decrease}, we obtain
\begin{align}
&V_o(\xi_{t+L})-V_o(\xi_t) =
\sum_{i=0}^{L-1}
\big(
V_o(\xi_{t+i+1})-V_o(\xi_{t+i})
\big) \nonumber\\
&\le
\sum_{i=0}^{L-1}
\left(
-\epsilon_o\|\xi_{t+i}\|^2
+\frac{1}{2}\|y_{t+i}\|_Q^2
+\|\bar u^*_{i|t}\|_R^2
\right).
\label{eq:online-uioss-sum}
\end{align}
Since $V_o(\xi_{t+L})\ge 0$, it follows that
\begin{align}
&\sum_{i=0}^{L-1}\epsilon_o\|\xi_{t+i}\|^2 \nonumber\\
&\le
V_o(\xi_t)
+
\sum_{i=0}^{L-1}
\left(
\|\bar u^*_{i|t}\|_R^2
+\frac{1}{2}\|y_{t+i}\|_Q^2
\right) \nonumber\\
&\le
\hat Y_L(\xi_t)
+
2\lambda_{\max}(Q)
\sum_{i=0}^{L-1}
\|(\bar y^*_{i|t}+h^*_{i|t})-y_{t+i}\|^2 \nonumber\\
&\quad+
2\lambda_{\max}(Q)
\sum_{i=0}^{L-1}
\|h^*_{i|t}\|^2 
\label{eq:online-state-bound-raw}
\end{align}
where $\ell^*_{i|t}:=\ell(\bar u^*_{i|t},\bar y^*_{i|t})$.

It remains to bound the two additional terms in
\eqref{eq:online-state-bound-raw}.
Let $y_{j,t+i}$ denote the output at time $t+i$ generated by the frozen system at $s_j$ (i.e. system \eqref{eq:sys-tv-sj}) under the same initial condition $\xi_t$ and input sequence $\bar u^*_{[0,L-1]|t}$.
Then,
\begin{align}
&\|(\bar y^*_{i|t}+h^*_{i|t})-y_{t+i}\|^2 \nonumber\\
&\le
2\|(\bar y^*_{i|t}+h^*_{i|t})-y_{j,t+i}\|^2
+
2\|y_{j,t+i}-y_{t+i}\|^2 .
\label{eq:online-output-split}
\end{align}
The first term on the right-hand side of \eqref{eq:online-output-split} is bounded by Assumption \ref{as:est-error-bound-tv}, since
$\bar y^*_{[0,L-1]|t}+h^*_{[0,L-1]|t}=Y_{F,s_j}g_t^*$ is generated by  \eqref{eq:kernel-estimator-tv}.
The second term is bounded, on the sublevel set
$\{\xi:\hat Y_L(\xi)\le \hat Y\}$, by a $\mathcal K_\infty$ function of $L_b$ due to Assumption \ref{as:tv-bound}.
Moreover, $\sum_{i=0}^{L-1}\|h^*_{i|t}\|^2
\le
\frac{\hat d}{\lambda_h}\hat J_L^*(\xi_t)
\le
\frac{\hat d}{\lambda_h}\hat Y$.
Consequently, there exists a function
$\hat\alpha_1(\cdot)\in\mathcal K_\infty$, depending on the fixed constants
$L$, $T_0$, and $\hat Y$, such that
\begin{equation}
\label{eq:online-state-bound}
\sum_{i=0}^{L-1}\epsilon_o\|\xi_{t+i}\|^2
\le
\hat Y_L(\xi_t)+\hat\alpha_1(\hat d).
\end{equation}

Inequality \eqref{eq:online-state-bound} is the online counterpart of \eqref{eq:alpha_1} in the proof of Theorem \ref{thm:stab-ideal}, with the nominal kernel error bound $\bar d$ replaced by the unified online perturbation scale $\hat d$.
Thus, for sufficiently large $L$ and sufficiently small $\hat d_0$, \eqref{eq:online-state-bound} ensures the existence of an index $i_\xi\in\{1,\ldots,L-1\}$ such that $\xi_{t+i_\xi}$ lies in the local region specified by Assumption \ref{as:stab-tv}.

At time $t+1$, let $j^+=\lfloor (t+1)/T_0\rfloor$ and $s_{j^+}=j^+T_0$.
The candidate input sequence is constructed by shifting the optimal sequence at time $t$ up to the index $i_\xi-1$ and appending the local controller $\kappa_{j^+}$ afterwards.
The corresponding $g_{t+1}$ and $h_{t+1}$ are constructed using  $\mathcal{D}_{s_{j^+}}^{c}$.
With this candidate solution, the same decomposition as in the proof of Theorem \ref{thm:stab-ideal} gives
\begin{align}
\hat J_L^*(\xi_{t+1})
&\le
-\ell^*_{0|t}
+
\hat J_L^*(\xi_t)
+
\hat I_1+\hat I_2+\hat I_3,
\label{eq:online-shifted-cost}
\end{align}
where $\hat I_1
:=
\sum_{i=0}^{i_\xi-2}
\left(
\|y_{t+i+1}\|_Q^2
-
\|\bar y^*_{i+1|t}\|_Q^2
\right)$, $\hat I_2
:=
\sum_{i=i_\xi-1}^{L-1}
\left(
\|\bar u_{i|t+1}\|_R^2
+
\|y_{t+i+1}\|_Q^2
\right)$, and $\hat I_3
:=
\frac{\lambda_h}{\hat d}\|h_{t+1}\|^2
+
\lambda_g\hat d\|g_{t+1}\|^2$.
The term $\hat I_1$ is bounded by a $\mathcal K_\infty$ function of $\hat d$ using Assumptions \ref{as:tv-bound} and \ref{as:est-error-bound-tv}, since it measures the mismatch between the shifted predicted output and the actual time-varying output.
For $\hat I_2$, Lemma \ref{lem:local-cost-bound-tv} and the choice of $i_\xi$ imply that $\hat I_2
\le
\frac{\lambda_{\max}(Q,R)\bar c_0\gamma_{\hat Y}}
{\epsilon_o(L-1)(1-\rho_s)}
\|\xi_t\|^2
+
\hat\alpha_2(\hat d)$
for some $\hat\alpha_2\in\mathcal K_\infty$.
Moreover, by Assumptions \ref{as:tv-bound} and \ref{as:est-error-bound-tv}, together with the boundedness of the kernel, there exists $\hat\alpha_3\in\mathcal K_\infty$ such that $\hat I_3\le \hat\alpha_3(\hat d)$.
Substituting these bounds into \eqref{eq:online-shifted-cost} yields
\begin{align}
\hat J_L^*(\xi_{t+1})
&\le
-\ell^*_{0|t}
+\!
\hat J_L^*(\xi_t)
+\!
\frac{\lambda_{\max}(Q,R)\bar c_0\gamma_{\hat Y}}
{\epsilon_o(L-1)(1-\rho_s)}
\|\xi_t\|^2
+ \!
\hat\alpha_4(\hat d),
\label{eq:online-J-decrease}
\end{align}
where $\hat\alpha_4\in\mathcal K_\infty$.

On the other hand, applying \eqref{eq:uioss-tv:decrease} for one step gives $V_o(\xi_{t+1})-V_o(\xi_t)
\le
-\epsilon_o\|\xi_t\|^2
+
\ell^*_{0|t}
+
\hat\alpha_5(\hat d)$
for some $\hat\alpha_5\in\mathcal K_\infty$, where $\hat\alpha_5$ accounts for the mismatch between the predicted and the actual output at time $t$.
Combining this inequality with \eqref{eq:online-J-decrease}, we obtain $\hat Y_L(\xi_{t+1})-\hat Y_L(\xi_t)
\le
-\hat a_L\epsilon_o\|\xi_t\|^2
+
\hat\alpha_{\hat Y}(\hat d)$, where $\hat a_L
:=
1-
\frac{\lambda_{\max}(Q,R)\bar c_0\gamma_{\hat Y}}
{\epsilon_o^2(L-1)(1-\rho_s)}$
is positive for all sufficiently large $L$, and
$\hat\alpha_{\hat Y}:=\hat\alpha_4+\hat\alpha_5\in\mathcal K_\infty$.
Finally, for sufficiently small $\hat d_0$, the above inequality implies that
$\hat Y_L(\xi_t)\le\hat Y$ leads to $\hat Y_L(\xi_{t+1})\le\hat Y$.
Recursive feasibility and \eqref{eq:online-lyap-ineq:decrease} then follow by induction.
This completes the proof.
\end{proof}


\begin{remark}[Online DDKPC under noisy measurement]
    When the available outputs are corrupted by bounded measurement noise (i.e., $\|n_t\| \le \bar{n}$), the same robustification principle used in Section \ref{sec:discussion-noise} applies here as well.
Specifically, one may choose $\hat{d} \ge \bar{d}_{on} + L_f + L_b + \bar{n}$.
The resulting robustified online DDKPC formulation is obtained by modifying constraint \eqref{eq:ddkmpc-online-ini} with \eqref{eq:ddkmpc-noise:ini}.
Under this modification, Lemma \ref{lem:local-cost-bound-tv} and Theorem \ref{thm:online-stability} still holds.
\end{remark}

\color{black}
\section{Numerical Examples}
We evaluate the proposed DDKPC framework on two simulation examples. The first is a nonlinear one-dimensional vehicle benchmark, where the proposed schemes are tested in noise-free, noisy, and slowly time-varying settings. 
This example is also used to compare the closed-loop performance of different numerical solvers. 
The second example is a PyBullet simulation of an A1 quadruped robot carrying an off-center payload, which demonstrates the applicability of DDKPC to a nonlinear and contact-rich locomotion system. 
The one-dimensional vehicle simulations are implemented in MATLAB 2022a, whereas the quadruped simulations are implemented in Python using PyBullet. 
All experiments are conducted on a laptop equipped with a 14-core Intel i7-12700H processor at 2.3 GHz.

\subsection{One-dimensional vehicle model}
\label{sec:sim:one}
\subsubsection{Simulation setup}
We consider a one-dimensional vehicle model governed by the following
continuous-time dynamics $m_0 \ddot{d} = u - c_1 \dot{d} - f_{nl}(\dot{d}, d)$,
where $m_0$ denotes the vehicle mass, $u$ is the control input, $c_1$ is
the viscous friction coefficient, and $f_{nl}(\dot{d}, d)$ represents the
nonlinear frictional forces arising from aerodynamic drag and road resistance.
For simplicity, we set $f_{nl}(\dot{d}, d)=c_2\dot d^2$, where $c_2$ denotes
the quadratic drag coefficient.
For the nominal test in Section~\ref{sec:sim_nominal}, $c_1$
and $c_2$ are kept constant, whereas the time-varying and noisy
time-varying tests in Section~\ref{sec:sim_tv} replace them by
$c_{1,t}$ and $c_{2,t}$.

Applying Euler discretization and defining the velocity $v = \dot{d}$ yield
the discrete-time state-space model
\begin{subequations}\label{eq:exm_car}
\begin{align}
    d_{t + 1} &= d_t + T_s v_t, \\
    v_{t + 1} &= v_t + {T_s}/{m_0} \big( u_t - c_1 v_t - c_2 v_t^2 \big), \\
    y_t &= d_t,
\end{align}
\end{subequations}
where $T_s = 0.02$\,s is the sampling period. The physical parameters are set
to $m_0 = 1$\,kg, $c_1 = 0.5$, and $c_2 = 0.1$. These parameters and the
state-space model \eqref{eq:exm_car} are used only to generate experimental
input-output data and to simulate the true plant; they are not available to
the data-driven controller.

To obtain informative offline data, the open-loop system is excited by a
combination of pulse signals and zero-mean white noise with variance $0.01$,
and $T=600$ samples are collected. 
Across all scenarios, for obtaining the best response, the controller hyperparameters are fixed as follows:
the past trajectory length $\eta=3$, the prediction horizon $L=10$, and the
output weighting matrix $Q=10$. 
In the simulations, rather than the separable per-stage input penalty used
in the theoretical presentation, we adopt the stacked quadratic cost
$\sum_{i=0}^{L-1}\|\bar y_{i|t}-y^s\|_Q^2
+
\sum_{i=0}^{L-1}\|\bar u_{i|t}-u^s\|_{r_u}^2
+
\sum_{i=1}^{L-1}\|\bar u_{i|t}-\bar u_{i-1|t}\|_{r_{\Delta u}}^2$.
Equivalently, the input-related terms are represented by
$\|\bar{u}_{[0,L - 1]|t}-\mathbf 1_L\otimes u^s\|_{\mathbf R}^2$,
where $\mathbf R
=
r_u I_L+r_{\Delta u}D_L^\top D_L$ and $D_L :=
\smat{
1 & -1 &        &        & 0\\
  & 1  & -1     &        &  \\
  &    & \ddots & \ddots &  \\
0 &    &        & 1      & -1
}
\in\mathbb R^{(L-1)\times L}$.
The penalty scaling parameters are chosen as $\bar{d}=10^{-3}$,
$\lambda_h/\bar{d}=10^4$, and $\lambda_g \bar{d}=5\times 10^{-3}$.
The control input is constrained to
$\mathcal{U}:=\{u\mid -1\le u\le 1\}$, and the regularization factor is
$\gamma=0.1624$. Unless otherwise specified, DDKPC uses a Gaussian kernel
$K(\zeta_i,\zeta_j)=\sigma_f^2
\exp(-\|\zeta_i-\zeta_j\|^2/(2\ell^2))$ with
$\sigma_f=0.62$ and $\ell=0.8385$. For comparison, we also test a polynomial
kernel $K(\zeta_i,\zeta_j)=\sigma_f^2(0.5\zeta_i^\top\zeta_j+1)^\ell$ with
$\sigma_f=0.05$ and $\ell=10$.
In the following, we compare a general-purpose nonlinear programming (NLP) solver
(\textit{fmincon}) with two first-order alternatives: the projected gradient
descent (GD) method in \cite{huang2023robust} and the accelerated GD method in
\cite[eqns. (43)--(45)]{wang2025optimization}, referred to as FastGD.

\subsubsection{Nominal DDKPC}
\label{sec:sim_nominal}

We first verify Algorithm~\ref{alg:ddkmpc-ideal} using noise-free data.
The closed-loop system is simulated for $20$\,s to track the constant reference
$y^s=0.5$. 
We compare the linear DDPC baseline (yellow dotted line) with Gaussian-kernel
DDKPC (green dashed, orange dashed, and blue solid lines) and polynomial-kernel
DDKPC (purple dash-dotted line).
For the standard linear DDPC, the kernel-based constraint \eqref{eq:ddkmpc-ideal:ker} is replaced by the standard Hankel matrix formulation:
\begin{equation}
    \smat{
    H_{\eta + L}(u^{\rm d})\\
    H_{\eta + L}(y^{\rm d})
    } g_t = \smat{
    \bar{u}_{[0,L - 1]|t}\\
    \bar{y}_{[0,L - 1]|t} + h_{[0,L - 1]|t}
    }.
\end{equation}

\begin{figure}[t]
	\centering
	\includegraphics[width=\linewidth]{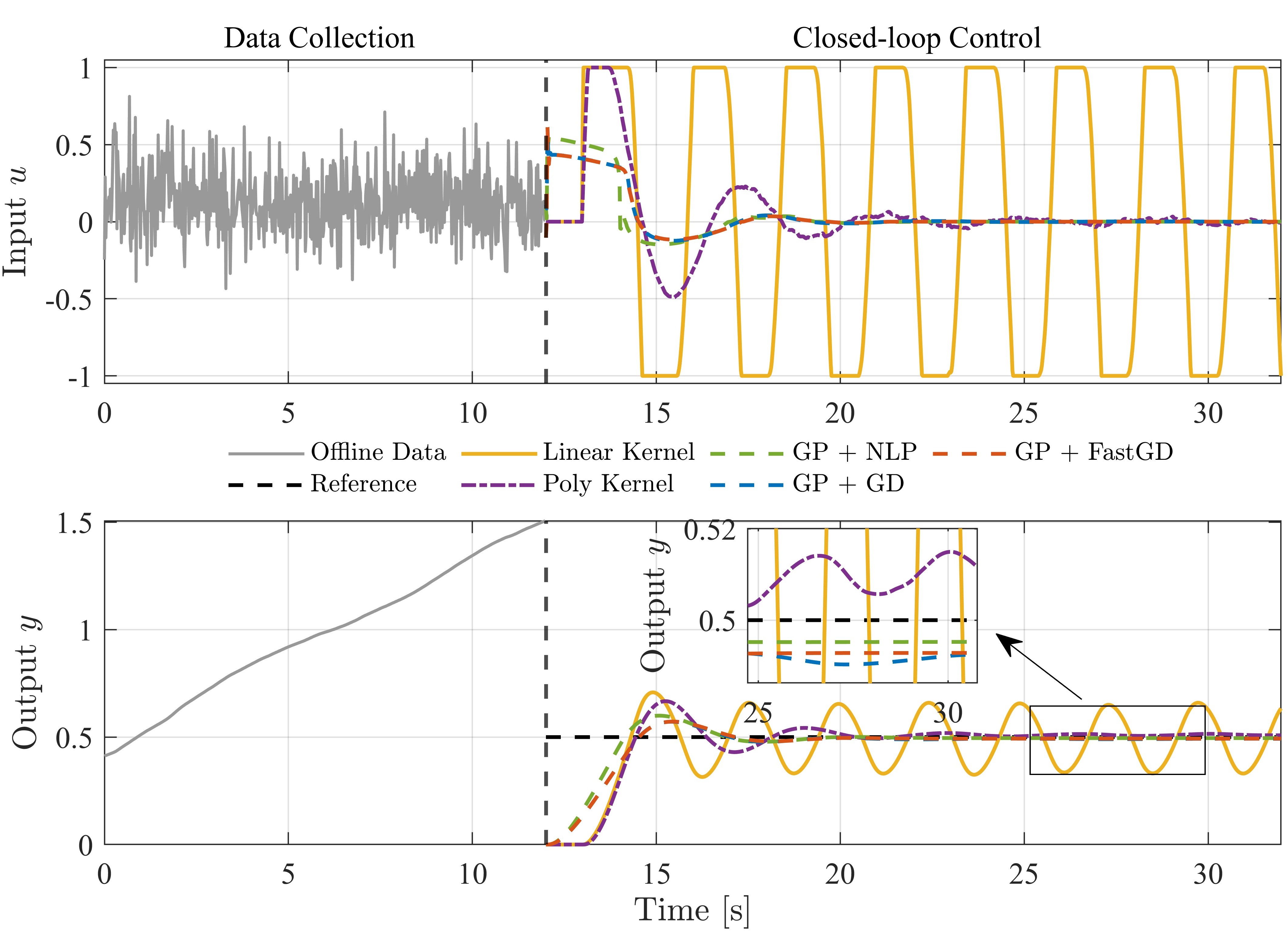}
	\caption{Closed-loop input and output trajectories under Algorithm \ref{alg:ddkmpc-ideal}, comparing different predictors and numerical solvers.}
    \label{fig:nominal}
\end{figure}

Fig.~\ref{fig:nominal} shows that the standard linear DDPC fails to stabilize
the nonlinear vehicle dynamics. In contrast, both DDKPC variants stabilize the
system and track the desired reference, consistent with the theoretical
guarantees in Section~\ref{sec:ddkmp:noise-free}. Although the nonlinearity in
\eqref{eq:exm_car} is polynomial, the Gaussian-kernel DDKPC attains slightly
better steady-state tracking accuracy than the polynomial-kernel variant,
possibly because multi-step prediction amplifies small fitting errors. For the
Gaussian-kernel DDKPC, the closed-loop tracking accuracy also depends on the
optimization solver, with the ordering NLP (\textit{fmincon}) $>$ FastGD
\cite{wang2025optimization} $>$ GD \cite{huang2023robust}. This trend is
consistent with Section~\ref{sec:discussion-computation}: higher numerical
precision leads to improved closed-loop performance.

\begin{figure}[t]
	\centering
	\includegraphics[width=\linewidth]{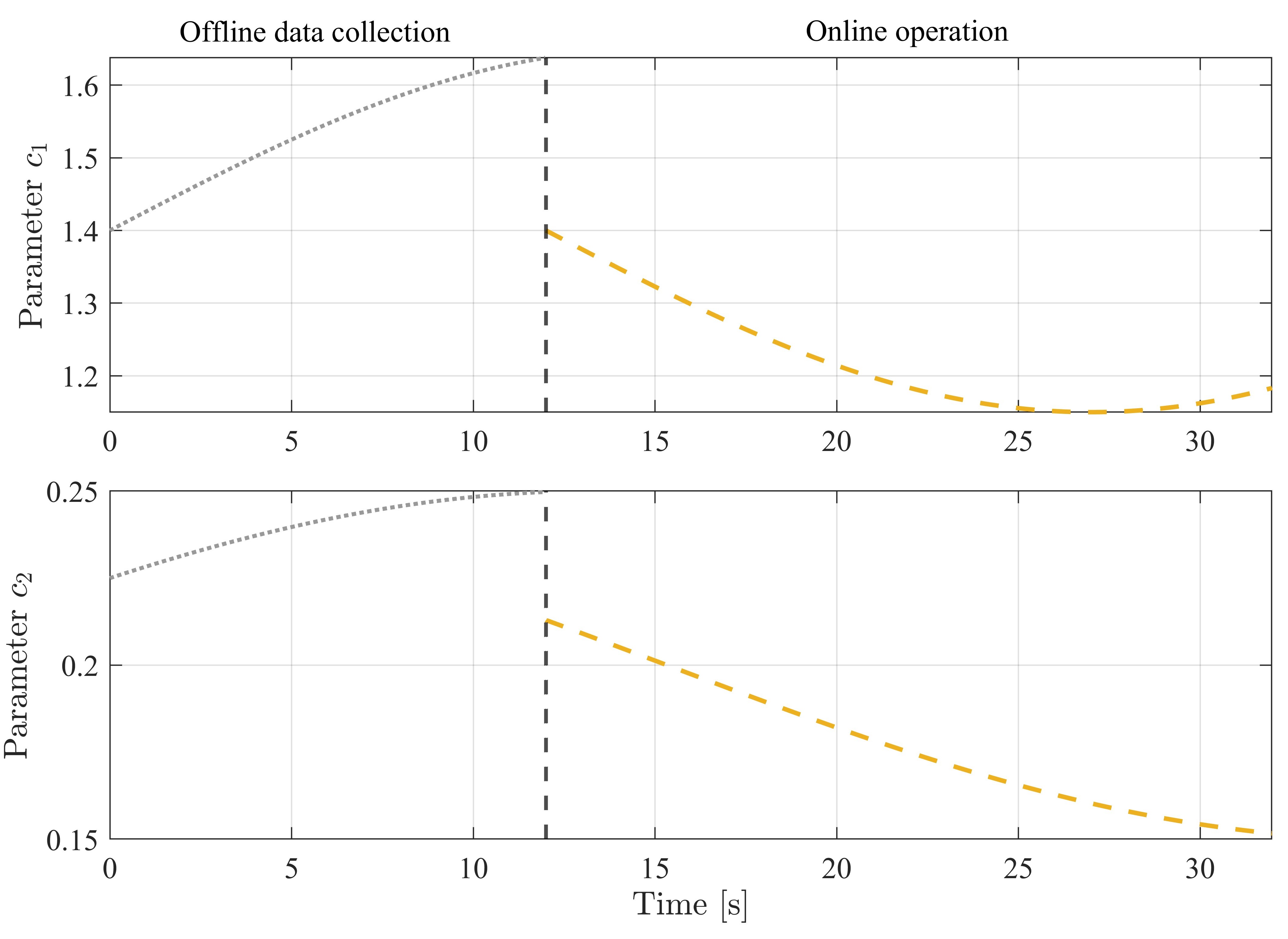}
	\caption{Time-varying parameters $c_1(t)$ and $c_2(t)$.}
    \label{fig:c1c2}
\end{figure}

\subsubsection{Online DDKPC under time-varying systems}
\label{sec:sim_tv}

We next evaluate Algorithm~\ref{alg:ddkmpc-online} on a slowly
time-varying variant of \eqref{eq:exm_car}, obtained by replacing
$c_1$ and $c_2$ with unknown coefficients $c_{1,t} =
\max\{c_{1,\min},
\bar c_1 + a_1 \sin({2\pi t}/{P_1})\}$ and $c_{2,t}=
\max\{c_{2,\min},
\bar c_2 + a_2 \sin({2\pi t}/{P_2}+\phi_2)\}$, where
$c_{1,\min} = 0.8$, $a_1 = 0.25$, $P_1 = 3000$,
$c_{2,\min} = 0.05$, $a_2 = 0.05$, $P_2 = 4000$, and
$\phi_2 = \pi/6$. Fig.~\ref{fig:c1c2} depicts the offline and
online evolutions of $c_1(t)$ and $c_2(t)$. The online trajectories
are separated from the nominal offline values, producing a gradual
mismatch between the static offline predictor and the current
input-output behavior. This mismatch motivates the periodic ALD-based update.

Since Section~\ref{sec:sim_nominal} shows that the Gaussian kernel gives better
tracking performance than the polynomial kernel, and since GD and FastGD attain
closed-loop performance comparable to the NLP solver with substantially lower
runtime, Algorithm~\ref{alg:ddkmpc-online} is evaluated using the Gaussian
kernel with the GD and FastGD solvers. 
Because the computational complexity of kernel methods grows rapidly with the dictionary size, we use a sparsified dictionary of $265$ elements selected from the full $600$-sample dataset used in Section~\ref{sec:sim_nominal}. This offline sparsification reduces the mean FastGD runtime from $55.10$ms to $23.63$ms.
Activating the online
dictionary update increases the mean FastGD runtime to $85.86$ms because of
dictionary maintenance and periodic predictor reconstruction, while preserving
a substantial computational advantage over the NLP baseline.

\begin{figure}[t]
    \centering
    \includegraphics[width=\linewidth]{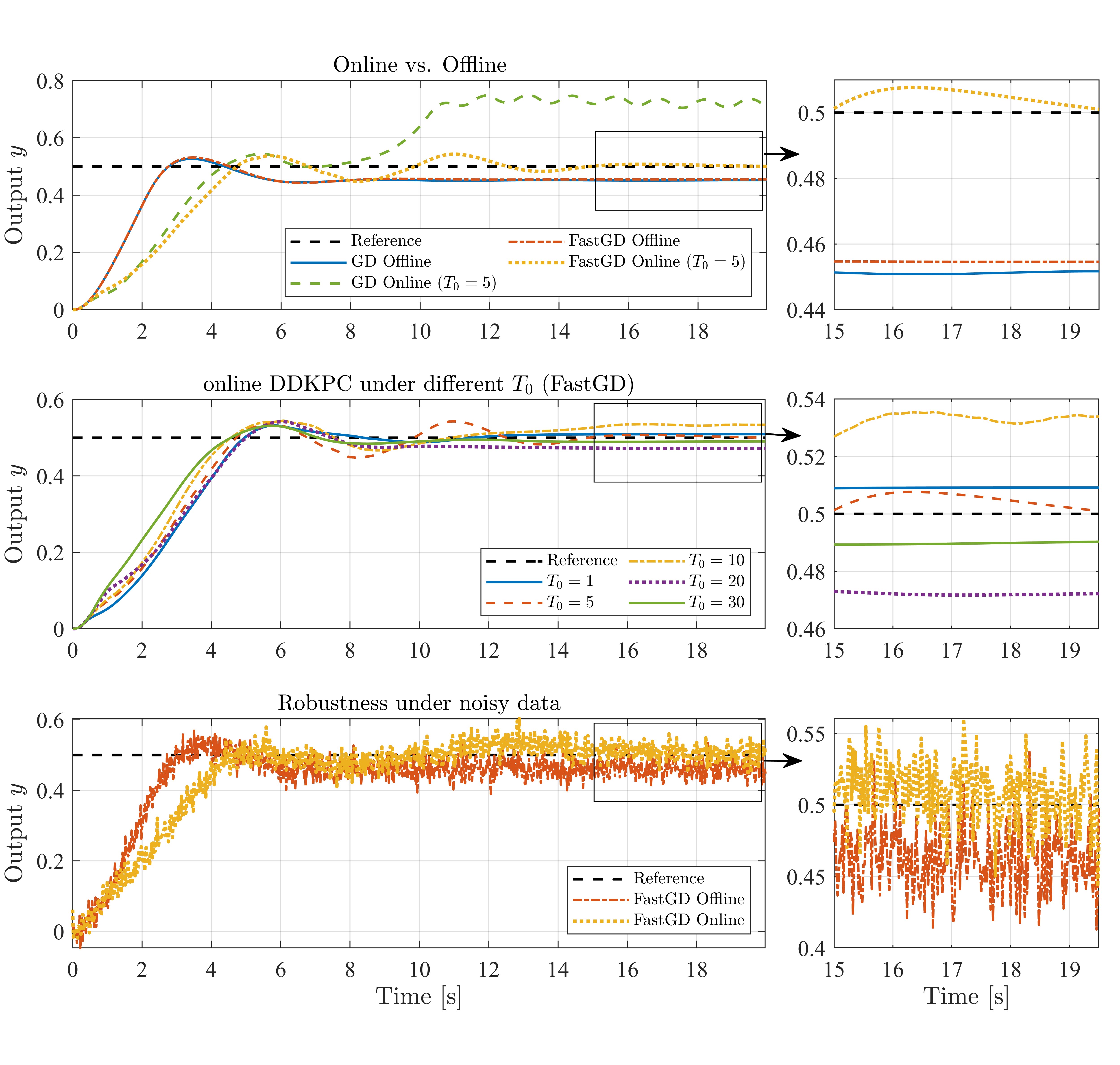}
    \caption{Closed-loop output trajectories for the time-varying system under
    Algorithm~\ref{alg:ddkmpc-online}. Top: static offline DDKPC vs. online
    DDKPC with GD and FastGD. Middle: FastGD online DDKPC with different update
    periods. Bottom: noisy-data test comparing static offline DDKPC with online
    FastGD DDKPC.}
    \label{fig:on_period}
\end{figure}

Fig.~\ref{fig:on_period} compares the closed-loop output trajectories under the
time-varying dynamics. In the top panel, the static offline configurations
exhibit a steady-state tracking bias because the sparsified dictionary is fixed
after offline data collection. With periodic online updates, FastGD reduces this
bias and brings the output closer to the reference. In contrast, the online GD
configuration does not improve upon its static counterpart. This behavior is
consistent with the lower numerical accuracy of standard GD, for which the
perturbations introduced by online predictor updates may offset the benefit of
using fresher data.

The middle panel of Fig.~\ref{fig:on_period} compares FastGD online DDKPC under
different update periods $T_0$. All tested periods
$T_0\in\{1,5,10,20,30\}$ preserve closed-loop stability, but the tracking
performance does not vary monotonically with $T_0$. This indicates that a
shorter update period is not necessarily preferable.
Therefore, the choice of $T_0$ reflects a nontrivial trade-off among tracking
accuracy, predictor variation, and computational cost, and its systematic
selection remains an important implementation issue.

The bottom panel of Fig.~\ref{fig:on_period} evaluates the noisy-data case with
$\bar n=0.001$. The online FastGD trajectory remains stable and tracks the
reference despite measurement noise, indicating that the robustified online
DDKPC can tolerate noisy data while adapting the dictionary to the slowly
time-varying plant.



\subsection{Velocity control of a quadruped robot}
\label{sec:sim-lateral}

We next consider a quadruped locomotion scenario in PyBullet to evaluate the
proposed DDKPC on a nonlinear, contact-rich system. As illustrated in
Fig.~\ref{fig:snapshot}(a), DDKPC is incorporated as an outer-loop residual
velocity controller around an existing nominal whole-body locomotion system,
while Fig.~\ref{fig:snapshot}(b) shows the corresponding PyBullet simulation
environment. A rigid payload mounted away from the trunk center changes the
closed-loop velocity response and produces persistent tracking errors. DDKPC
generates residual velocity commands for the nominal controller using only
measured input-output data.

\begin{figure}
    \centering
    \includegraphics[width=\linewidth]{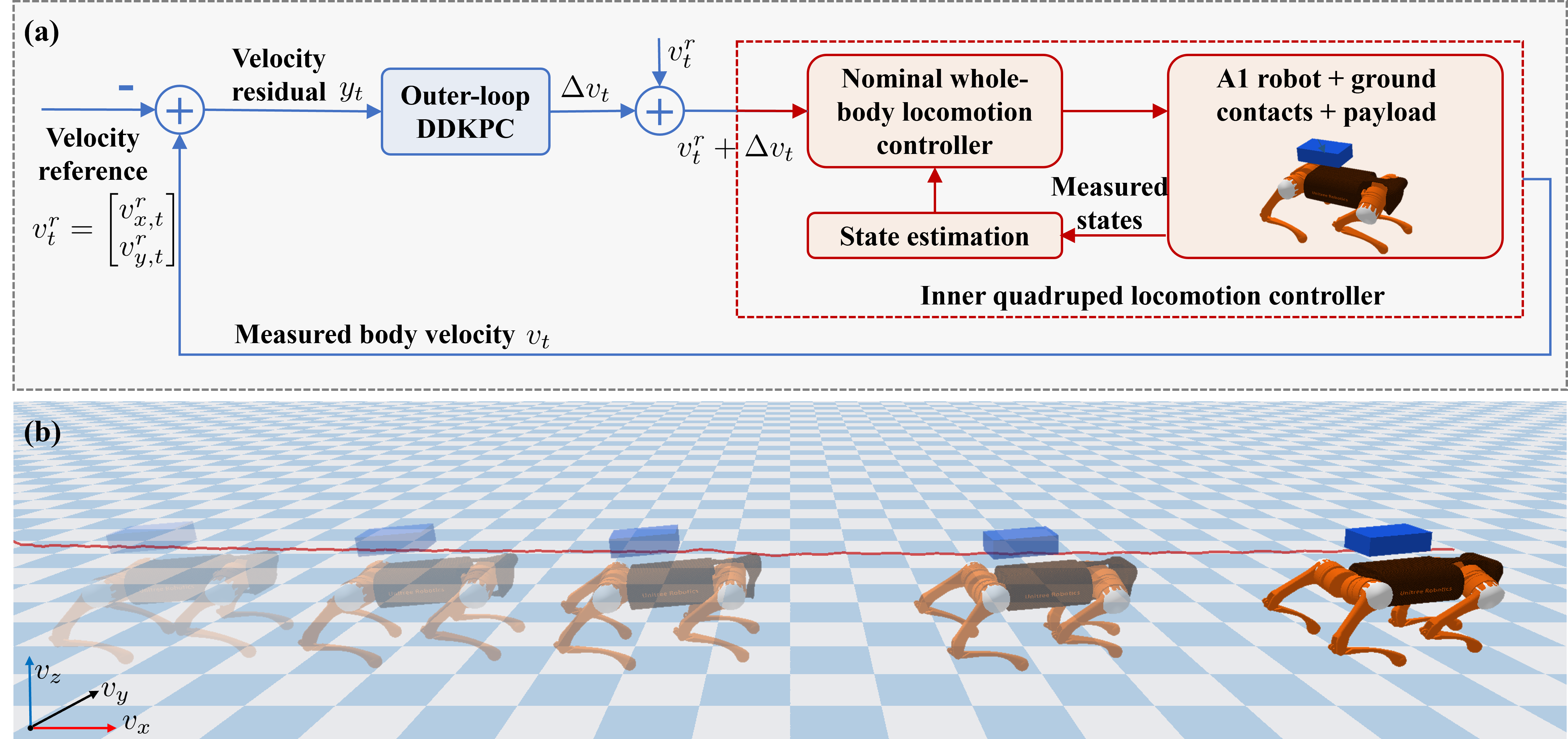}
    \caption{Velocity control architecture and PyBullet simulation
of the payload-carrying A1 quadruped.}
    \label{fig:snapshot}
\end{figure}

\subsubsection{Simulation setup}
\label{sec:sim-quad-setup}

The A1 is driven by an existing whole-body locomotion stack that maps desired
body velocities to joint commands \cite{da2021learning}. 
As shown in
Fig.~\ref{fig:snapshot}(b), the reference velocity
$ v_t^r=[v_{x,t}^r,v_{y,t}^r]^\top$ is augmented by the DDKPC residual command
$\Delta v_t=[\Delta v_{x,t},\Delta v_{y,t}]^\top$, yielding
$ v_t^{\rm cmd}= v_t^r+\Delta v_t$ for the nominal locomotion controller.
From the viewpoint of DDKPC, the nominal controller, state estimator, robot dynamics, intermittent ground contacts, and payload constitute the unknown closed-loop plant; no rigid-body or contact model is used.
For forward-velocity control, the DDKPC input and output are
$u_t=\Delta v_{x,t}$, $ y_t=v_{x,t}-v_{x,t}^{r}$.
For planar-velocity control, they are extended to $u_t=\smat{\Delta v_{x,t}\\ \Delta v_{y,t}}$, $
    y_t=\smat{v_{x,t}-v_{x,t}^{r}\\v_{y,t}-v_{y,t}^{r}}$.

A $2$\,kg box payload is rigidly attached $0.08$\,m behind and
$0.14$\,m above the trunk center, with no lateral offset. The resulting
rearward and upward shift of the overall center of mass modifies the
closed-loop velocity response of the nominal locomotion controller. 
The DDKPC sampling period is $T_s=0.1$\,s. 
For each
input-output interface, the fixed data dictionary is constructed from two
$90$\,s records obtained under the same reference profile, one without the
payload and one with the payload. 
This provides data from both nominal and
payload-loaded operating conditions without requiring a model of either
condition. The past trajectory length is $\eta=3$, and the first $3$\,s of
each closed-loop test are excluded from the performance evaluation.
We adopt the same Gaussian kernel formulation as that used in Section \ref{sec:sim:one}.
The dictionary and $\overline{\mathbf K}$ are constructed offline and kept fixed during
the closed-loop tests. 

\subsubsection{Forward-velocity tracking}
\label{sec:sim-quad-vx}

The forward-velocity reference is $0.4$\,m/s for $0\leq t<12$\,s,
$0.6$\,m/s for $12\leq t<28$\,s, and $0.4$\,m/s thereafter. We first solve
the nonlinear program in Algorithm~\ref{alg:ddkmpc-ideal}  using \texttt{scipy.optimize.minimize} with \texttt{method='L-BFGS-B'} in Python.
The prediction horizon is $L=3$, the dictionary contains $60$ elements, and
$(\sigma_f,\ell^2,\gamma)=(1,1.6,10^{-6})$. 
Here, the kernel hyperparameters are selected empirically to balance multi-step prediction accuracy, closed-loop tracking performance, and numerical conditioning.
The residual command is constrained
to $[-0.2,0.2]$\,m/s, with $Q=10$, $r_u=0.2$, and
$r_{\Delta u}=0.05$. Similar to Section~\ref{sec:sim_nominal}, we also use GD
to reduce the online computation. 

\begin{figure}[t]
    \centering
    \includegraphics[width=\linewidth]{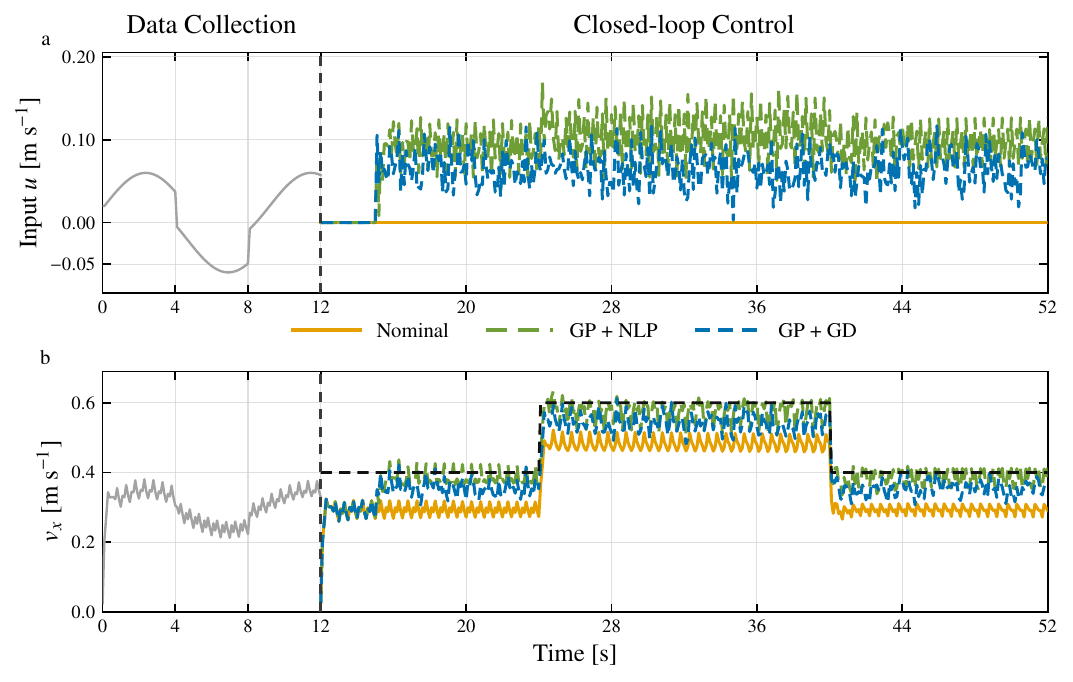}
    \caption{Forward-velocity tracking with a fixed $2$kg payload under the
    nominal whole-body controller, GD-DDKPC, and NLP-DDKPC. The reference velocity is represented by black dashed line.}
    \label{fig:quad-vxonly}
\end{figure}

As shown in Fig.~\ref{fig:quad-vxonly}, the payload causes a persistent velocity deficit under the nominal controller. The tracking RMSE, computed after excluding the first $3$s of closed-loop operation, is $0.1133$m/s for the nominal controller, $0.0351$m/s for NLP-DDKPC, and $0.0567$m/s for GD-DDKPC. Thus, NLP-DDKPC and GD-DDKPC reduce the RMSE by approximately $69.1\%$ and $49.9\%$, respectively. NLP-DDKPC achieves the highest tracking accuracy, whereas the GD implementation reduces the mean solution time from $82.66$ms to $23.43$ms.
We further compared Algorithms \ref{alg:ddkmpc-ideal} and \ref{alg:ddkmpc-online} under a payload-switching condition, where the reference velocity is fixed at $0.8$m/s and the payload changes from $0$ to $2$kg at $t=10$s and is removed at $t=30$s.
For this test, unlike the fixed-payload case above, the initial offline dictionary is constructed exclusively from payload-free data, so that the $2$kg condition is not represented in the training data. As shown in Fig.~\ref{fig:quad-offon}, the online DDKPC adapts more effectively to the payload variation. Over the complete test, the RMSEs of the nominal controller, offline DDKPC, and online DDKPC are $0.1128$, $0.1011$, and $0.0626$m/s, respectively. During the payload-loaded interval, the online scheme reduces the RMSE by $51.6\%$ and $38.1\%$ relative to the nominal controller and offline DDKPC, respectively.

 \begin{figure}[t]
    \centering
    \includegraphics[width=\linewidth]{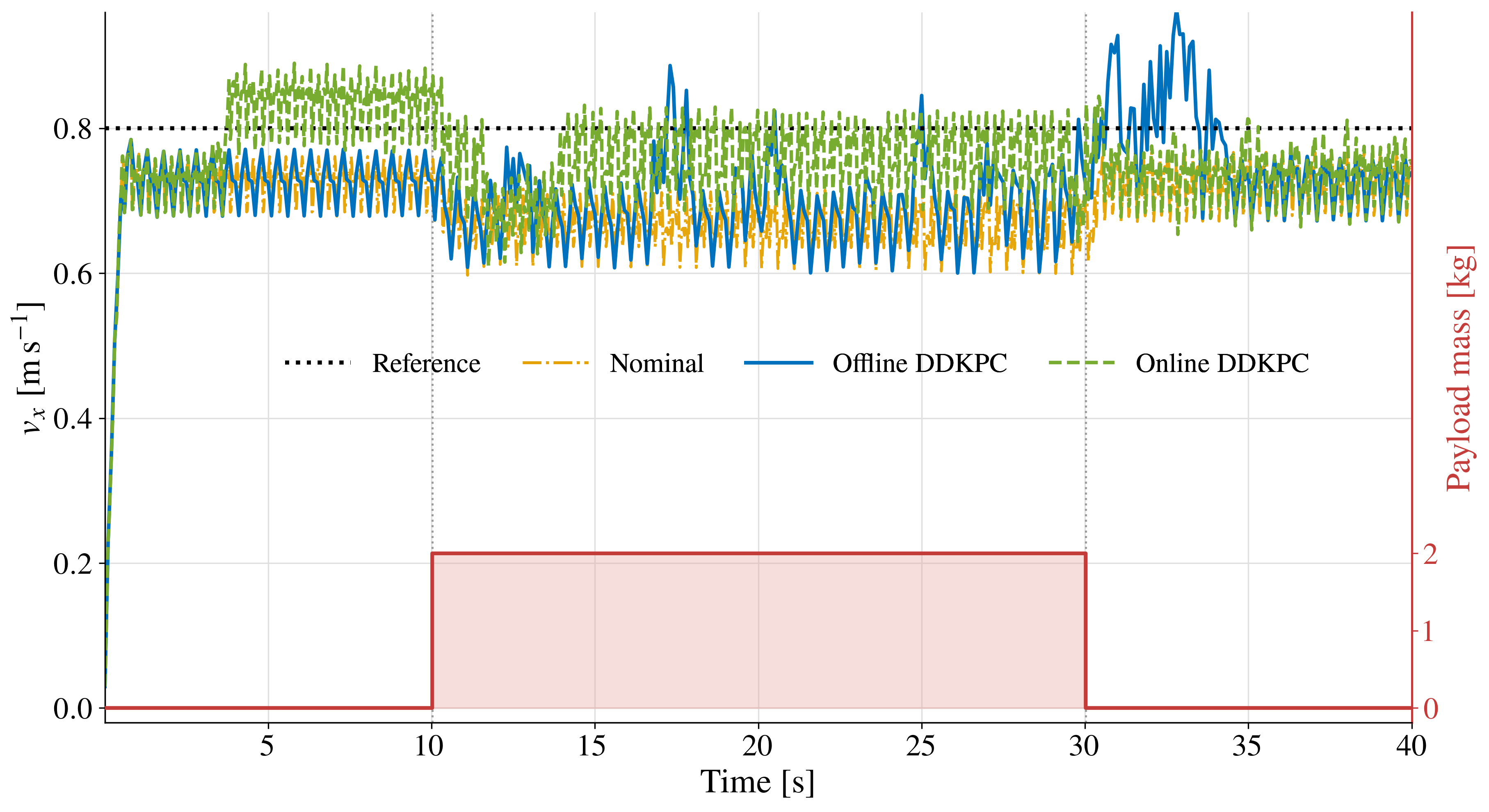}
    \caption{Forward-velocity tracking under a payload-switching condition: static offline DDKPC vs. online DDKPC.}
    \label{fig:quad-offon}
\end{figure}
\color{black}
\subsubsection{Planar-velocity tracking}
\label{sec:sim-quad-vxvy}

We next consider simultaneous forward- and lateral-velocity tracking. The
reference $(v_x^r,v_y^r)$ is $(0.45,0.15)$\,m/s for the first $12$\,s,
$(0.60,0.30)$\,m/s for the next $16$\,s, and $(0.45,0.15)$\,m/s for the final
$12$\,s. Since the computational cost of the nonlinear program increases
substantially with the input-output dimension, the planar test uses GD with
$L=6$ and a $300$-element dictionary selected by offline ALD. The kernel
parameters are $(\sigma_f,\ell^2,\gamma)=(0.9,1.2,0.1624)$, and the residual
commands satisfy $|\Delta v_x|,|\Delta v_y|\leq0.1$\,m/s.
Fig.~\ref{fig:quad-vxvy} shows that GD-DDKPC reduces the payload-induced
tracking errors in both channels. The raw RMSE of $v_x$ decreases from
$0.1179$ to $0.0624$\,m/s, while that of $v_y$ decreases from $0.1313$ to
$0.1173$\,m/s. The corresponding planar RMSE decreases from $0.1248$ to
$0.0939$\,m/s, representing a reduction of $24.7\%$.
These results show that Algorithm~\ref{alg:ddkmpc-ideal} can compensate for payload-induced velocity mismatch in a quadruped locomotion system. 

\begin{figure}[t]
    \centering
    \includegraphics[width=\linewidth]{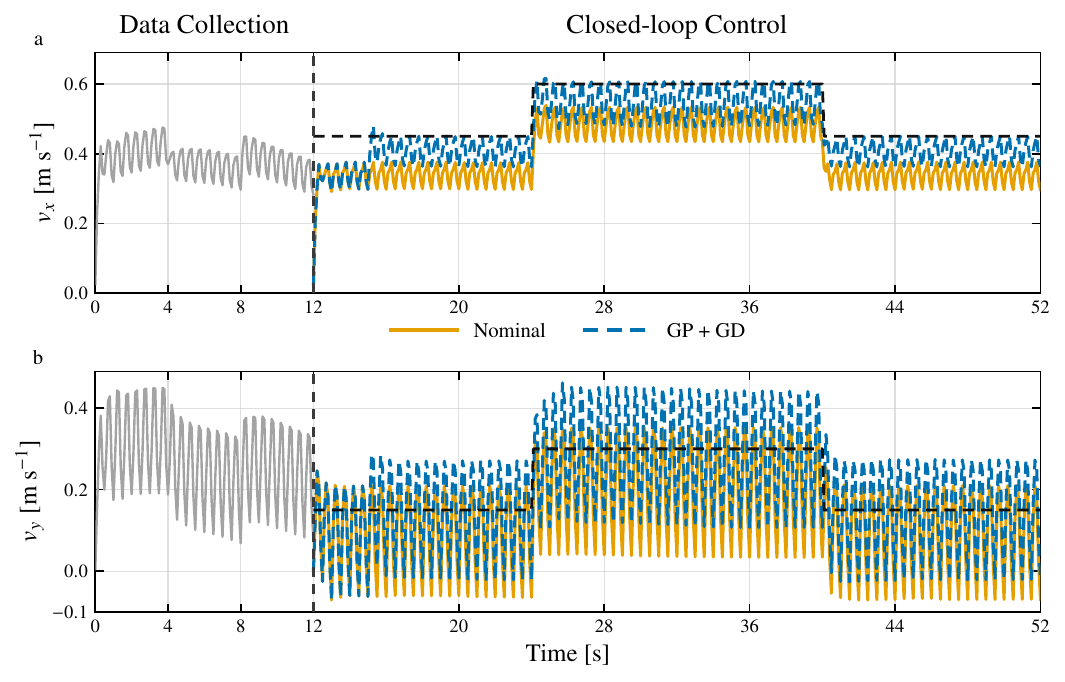}
    \caption{Planar body-velocity tracking with a fixed $2$\,kg payload under
    the nominal whole-body controller and GD-DDKPC.}
    \label{fig:quad-vxvy}
\end{figure}

\color{black}
\section{Conclusion}
\label{sec:conclusion}
This paper proposed a DDKPC of unknown nonlinear systems, relying exclusively on input-output trajectories. By integrating a robust predictive control architecture with an implicit multi-step predictor derived via the representer theorem, the proposed method effectively circumvents the need for explicit parametric modeling. Theoretical analysis established that, under nominal conditions, the scheme guarantees recursive feasibility and practical closed-loop stability, provided that the prediction horizon is adequately long and the initial kernel approximation error is bounded.
To bridge the gap between theoretical formulations and practical implementation, the framework was extended to address the computational demands of real-time deployment and the presence of measurement noise. 
Furthermore, the proposed DDKPC framework was extended to slowly time-varying nonlinear
systems by letting the data-dependent predictor in DDKPC evolve online.
The recursive feasibility and practical stability guarantees were shown to carry over under suitable regularity and prediction-error conditions.

Promising directions for future research include developing real-time adaptation mechanisms for the kernel hyperparameters and investigating theoretical approaches to relax the bounded perturbation assumption during online data updates.

\bibliographystyle{IEEEtran}
\bibliography{ddker}

\end{document}